\documentclass[a4paper,11pt]{article}

\usepackage[utf8]{inputenc}
\usepackage[british]{babel}
\usepackage[mono=false]{libertine}
\usepackage[T1]{fontenc}
\usepackage{amsthm}
\usepackage[top=2.5cm, bottom=2.5cm, left=2cm, right=2cm]{geometry}
\usepackage[font=small]{caption}
\usepackage{enumitem}
\usepackage{titlesec}
\usepackage[usenames, dvipsnames]{color}
\makeatletter
\newcommand{\setappendix}{Appendix~\thesection:~~}
\newcommand{\setsection}{\thesection~~}
\titleformat{\section}{\bfseries\LARGE}{%
	\ifnum\pdfstrcmp{\@currenvir}{appendices}=0
	\setappendix
	\else
	\setsection
\fi}{0em}{}
\makeatother
\usepackage[titletoc]{appendix}
\usepackage[pdftex]{graphicx}
\usepackage{epstopdf}
\usepackage{array}
\usepackage{url}
\usepackage{mathtools}
\usepackage{amssymb,amsfonts,amsmath}
\usepackage{dsfont}
\usepackage{stmaryrd}
\usepackage{bm}
\usepackage{footnote}
\usepackage[utf8]{inputenc} 
\usepackage[T1]{fontenc}    
\usepackage{booktabs}       
\usepackage{nicefrac}       
\usepackage{microtype}      
\usepackage{tikz}
\usepackage{subcaption}
\usepackage{algorithmic}
\usepackage{algorithm}
\usepackage{wrapfig}
\usepackage{xr}
\usepackage{xr-hyper}
\usepackage{braket}
\usepackage{cite}

\usepackage{mathtools}

\usepackage[pdfpagemode=UseNone,bookmarksopen=false,colorlinks=true,urlcolor=blue,citecolor=magenta,citebordercolor=blue,linkcolor=blue]{hyperref}

\def \({\left(}
\def \){\right)}
\def \[{\left[}
\def \]{\right]}
\newcommand{\nn}{\nonumber \\}

\newcommand{\sign}{\text{ sign}}

\newcommand{\be}{\begin{equation}}
\newcommand{\ee}{\end{equation}}
\newcommand{\beqa}{\begin{eqnarray}}
\newcommand{\eeqa}{\end{eqnarray}}

\newcommand{\bea}{\begin{align}}
\newcommand{\eea}{\end{align}}

\newtheorem{theorem}{Theorem}
\newtheorem{lemma}[theorem]{\textbf{Lemma}}

\newtheorem{proposition}[theorem]{\textbf{Proposition}}

\newtheorem{hypothesis}[theorem]{\textbf{Hypothesis}}

\DeclareMathAlphabet{\varmathbb}{U}{bbold}{m}{n}

\newcommand{\EE}{\mathbb{E}}

\usepackage{setspace}

\numberwithin{equation}{section}

\DeclareTextCommandDefault{\nobreakspace}{\leavevmode\nobreak\ } 

\begin{document}
\title{Overlap matrix concentration in optimal Bayesian inference}

\author{Jean Barbier}
\date{}
\maketitle
{\let\thefootnote\relax\footnote{
\!\!\!\!\!\!\!\!\!\!\!The Abdus Salam International Center for Theoretical Physics, Trieste, Italy.\\
{\it jbarbier@ictp.it}
}}
\setcounter{footnote}{0}

\begin{abstract}
We consider models of Bayesian inference of signals with vectorial components of finite dimensionality. We show that under a proper perturbation these models are replica symmetric in the sense that the overlap matrix concentrates. The overlap matrix is the order parameter in these models and is directly related to error metrics such as minimum mean-square errors. Our proof is valid in the optimal Bayesian inference setting. This means that it relies on the assumption that the model and all its hyper-parameters are known so that the posterior distribution can be written exactly. Examples of important problems in high-dimensional inference and learning to which our results apply are low-rank tensor factorization, the committee machine neural network with a finite number of hidden neurons in the teacher-student scenario, or multi-layer versions of the generalized linear model.
\end{abstract}

\section{Introduction}
This decade is witnessing a burst of mathematical studies related to high-dimensional inference and learning problems. One reason is that an important arsenal of methods, developed in particular by the physicists and mathematicians working on the rigorous aspects of spin glasses, has found a new rich playground where it can be applied with success \cite{GhirlandaGuerra:1998,guerra2002thermodynamic,guerra2003broken,talagrand2010meanfield1,talagrand2010meanfield2,talagrand2006parisi,panchenko2013sherrington}. Models in learning like the perceptron and Hopfield neural networks have been analyzed in depth since the eighties by the physics community \cite{hopfield1982neural,gardner1988optimal,gardner1989three,mezard1989space,gyorgyi1990first,seung1992statistical}, or in inference, e.g., in the context of communications and error correcting codes \cite{sourlas1989spin,Tanaka_CDMA}, using powerful but non-rigorous techniques such as the replica and cavity methods \cite{mezard1987spin,mezard2009information}. But due to the difficulty and richness of these models rigorous results experienced some delay with respect to (w.r.t.) the physics appoaches and were restricted to very specific models such as the famous Sherrington-Kirkpatrick model \cite{guerra2002thermodynamic,talagrand2006parisi,panchenko2013sherrington}. The trend is changing and it is fair to say that the gap between heuristic (yet often exact) physics approaches and rigorous ones is quickly shrinking. In particular important progress towards the vindication of the replica and cavity methods has been made recently in the context of high-dimensional Bayesian inference and learning. Examples of problems in this class where the physics approaches are now rigorously settled include low-rank matrix and tensor factorization \cite{korada2009exact,deshpandePCA,deshpande2015asymptotic,krzakala2016mutual,XXT,barbier2018rank,MiolaneXX,MiolaneUV,MiolaneTensor,2017arXiv170910368B,BarbierM17a,2018arXiv180101593E,mourrat2018hamilton,Luneau_matFacto,mourrat2019hamilton}, random linear and generalized estimation \cite{KoradaMacris_CDMA,barbier_allerton_RLE,barbier_ieee_replicaCS,reeves2016replica,barbier2017phase,RLEStructuredMatrices}, models of neural networks in the teacher-student scenario \cite{barbier2017phase,SM_Arxiv_Aubin2018,Gabrie:NIPS2018}, or sparse graphical models such as error-correcting codes and block models \cite{coja2017information,abbe2018community,BarCM:2018}. 

All these results are based in some way or another on the control of the fluctuations of the order parameter of the problem, the overlap, which quantifies the quality of inference. Optimal Bayesian inference --optimal meaning that the true posterior is known-- is an ubiquitous setting in the sense that the overlap can be shown to concentrate, and this in the whole regime of parameters (amplitude of the noise, number of observations/data points divided by the number of parameters to infer etc). 
When the overlap is self-averaging (which is the case in optimal Bayesian inference under a proper perturbation, see Theorems~\ref{prop:thermalBOund} and \ref{thm:Q_con}) then one expects {\it replica symmetric} variational formulas for the asymptotic free energy or mutual information density, as understood a long time ago by physicists \cite{Pastur-Shcherbina-1991,Pastur-Shcherbina-Tirozzi-1994}. Actually in the physics literature replica symmetry is generally the term used to precisely mean that the order parameter concentrates. This is in contrast with models where the overlap is not self-averaging, like in spin glasses at low temperature or combinatorial optimization problems, which leads to more complicated formulas for the free energy computed using Parisi's {\it replica symmetry breaking} scheme \cite{Parisi-1980,mezard1987spin,mezard2009information,talagrand2010meanfield2,talagrand2006parisi,panchenko2013sherrington}. 

In most of the studied statistical models the overlap order parameter is a scalar. In the context of optimal Bayesian inference it is now quite standard to show that when the overlap is a scalar it is self-averaging in the whole phase diagram, see, e.g., \cite{2019arXiv190106516B,BarbierM17a}. The techniques to do so have been developed in the context of communications starting with \cite{macris2007griffith,MacrisKudekar2009,KoradaMacris_CDMA} (and then generalized in \cite{andrea2008estimating,coja2017information}), and are extensions of methods used in the analysis of spin glasses \cite{GhirlandaGuerra:1998,aizenman1998stability,contucci2005spin,contucci2007ghirlanda,contucci2013perspectives,talagrand2010meanfield1,talagrand2010meanfield2,panchenko2010ghirlanda}. In this paper we consider instead Bayesian inference problems where the signal to be reconstructed is made of vectorial components. In this case the overlap is a matrix and the associated replica formulas are variational formulas over matrices. The concentration techniques developed for scalar overlaps do not apply directly, and need to be extended using new non-trivial ideas. In particular, new difficulties will appear w.r.t. the scalar case due to the fact that overlap matrices are not symmetric objects. Examples of problems where matrix overlaps appear are the factorization of matrices and tensors of rank greater than one \cite{MiolaneXX}, or the so-called committee machine neural network with few hidden neurons \cite{schwarze1992generalization,monasson1995weight,engel2001statistical,SM_Arxiv_Aubin2018}. In the context of spin glasses, matrix overlap order parameters have also appeared recently in studies of vectorial versions of the Potts and mixed $p$-spin models by Panchenko \cite{panchenko2018potts,panchenko2018free}; in these models replica symmetry breaking occurs and the overlap does not concentrate. Let us also mention the recent work by Agliari and co-workers \cite{agliari2018non} on a ``multi-species'' version of the Hopfield model, where a matrix order parameter also appears. There concentration of overlap, in the replica symmetric region where concentration is expected, is assumed based on strong physical arguments. In the context of optimal Bayesian inference the situation is more favorable than in spin glasses: thanks to special identities that follow from Bayes' rule and known as ``Nishimori identities'' in statistical physics (see, e.g., \cite{nishimori2001statistical,contucci2009spin}), we show in this paper how to control the overlap fluctuations in the whole phase diagram\footnote{Let us mention another particular setting where the (scalar) overlap as well as its multi-body generalizations (which appear in diluted problems, i.e., problems defined by sparse graphical models) can be controlled under proper perturbations for all the parameters values in the phase diagram: ferromagnetic models \cite{2019arXiv190106521B}.}.

Section~\ref{sec2} presents the general setting, gives a few examples of models covered by our results, and explains the important Nishimori identity for optimal Bayesian inference problems. In section~\ref{sec:Gauss_channel} we introduce the perturbation needed in order to prove overlap concentration, and then give our main results Theorems~\ref{prop:thermalBOund} and \ref{thm:Q_con}. Then in section~\ref{sec4} we provide the proof of Theorem~\ref{thm:Q_con}. Finally in section~\ref{app:Lconc} we prove an important intermediate concentration result for another matrix, that will be key in controlling the overlap.
\section{Optimal Bayesian inference of signals with vector entries}\label{sec2}
\subsection{Setting}\label{sec:BayesOptInf}
Consider a model where a signal $X=(X_{ik})\in [-S,S]^{n\times K}$ made of $n$ components (indexed by $i$), that are each a $K$-dimensional bounded real vector (with dimensions indexed by $k$), is generated probabilistically. Its probability distribution $P_0$, called prior, may depend on a generic hyper-parameter $\theta_0\in \Theta_0$ with $\Theta_0$ an arbitrary real set, i.e., $$X\sim P_0(\,\cdot\,|\theta_0)\,.$$ We assume that the prior has bounded support (with $S<+\infty$ arbitrarily large but independent of $n$). Then some real data (also called observations) $\widetilde Y$ are generated conditionally on the unknown signal $X$ and an hyper-parameter $\theta_{\rm out}$ belonging to a generic real set $\Theta_{\rm out}$. Namely, the data $$\widetilde Y \sim P_{\rm out}(\,\cdot\,| X,\theta_{\rm out})\,,$$ with $\widetilde Y\in \widetilde{\cal Y}$ a generic real set: the data $\widetilde Y$ and hyper-parameters $\theta_0$, $\theta_{\rm out}$ can be real numbers, vectors, tensors etc. The conditional distribution $P_{\rm out}$ is called likelihood, or ``output channel''. We also assume that the hyper-parameters $\theta_0$ and $\theta_{\rm out}$ are also probabilistic, with respective probability distributions $P_{\theta_0}$ supported on $\Theta_0$, and $P_{\theta_{\rm out}}$ supported on $\Theta_{\rm out}$. This formulation includes the case of deterministic hyper-parameters choosing Dirac delta measures $P_{\theta_{\rm out}}=\delta_{\theta_{\rm out}}$ and $P_{\theta_0}=\delta_{\theta_0}$.

The inference task is to recover the signal $X$ as accurately as possible given the data $\widetilde Y$. We moreover assume that the hyper-parameters $\theta\equiv(\theta_0, \theta_{\rm out})$, the likelihood $P_{\rm out}$ and the prior $P_0$ are known to the statistician, and call this setting {\it optimal} Bayesian inference.

The information-theoretical optimal way of reconstructing the signal follows from its posterior distribution. Using Bayes' formula the posterior reads
\begin{align}
P(X=x|\widetilde Y,\theta)=P(x|\widetilde Y,\theta)&=\frac{P_0(x|\theta_0) P_{\rm out}(\widetilde Y|x,\theta_{\rm out})}{\int dP_0(x'|\theta_0) P_{\rm out}(\widetilde Y|x',\theta_{\rm out})} \nn
&= \frac{1}{{\cal Z}_{0,n}(\widetilde Y,\theta)} P_0(x|\theta_0)\exp\{-{\cal H}_0(x,\widetilde Y,\theta_{\rm out})\}\,. \label{post}
\end{align}
Employing the language of statistical mechanics we call 
\begin{align*}
{\cal H}_0(x,\widetilde Y,\theta_{\rm out})\equiv -\ln P_{\rm out}(\widetilde Y|x,\theta_{\rm out})
\end{align*}
the base {\it Hamiltonian}, while the posterior normalization ${\cal Z}_{0,n}(\widetilde Y,\theta)$ is the {\it partition function} of the base inference model. Finally the averaged {\it free energy} is minus the averaged log-partition function:
\begin{align*}
f_{0,n}\equiv 	-\frac1n \EE\ln{\cal Z}_{0,n}(\widetilde Y,\theta) =-\frac1n \EE\ln  \int dP_0(x|\theta_0)\exp\{-{\cal H}_0(x,\widetilde Y,\theta_{\rm out})\}\,.
\end{align*}
The average $\EE=\EE_{\theta}\EE_{X|\theta_0}\EE_{\widetilde Y|X,\theta_{\rm out}}$ is over the randomness of $(\theta,X,\widetilde Y)$. These are jointly called the {\it quenched variables} as they are fixed by the realization of the problem, in contrast with the dynamical variable $x$ which fluctuates according to the posterior. In general $\EE$ will be used for an average w.r.t. all random variables in the ensuing expression. Note that the averaged free energy is nothing else than the Shannon entropy density of the observations (given the hyper-parameters): $f_{0,n}=\frac1nH(\widetilde Y|\theta)$. Therefore it is simply related to the mutual information density between the observations and the signal:
\begin{align*}
	\frac1n I(X;\widetilde Y|\theta)= f_{0,n}-\frac1n H(\widetilde Y|X,\theta)\,.
\end{align*}
The conditional entropy $\frac1n H(\widetilde Y|X,\theta)$ is often easy to compute, as opposed to the averaged free energy. 

We call model \eqref{post} the ``base model'' in contrast with the perturbed model presented in section~\ref{sec:Gauss_channel}, a slightly modified version of the base model where additional side-information is given, and for which overlap concentration can be proved without altering the thermodynamic $n\to +\infty$ limit of the averaged free energy (if it exists), see Lemma~\ref{lemma:same_f}.

The central object of interest is the $K\times K$ {\it overlap matrix} (or simply overlap) $Q=(Q_{kk'})$ defined as
\begin{align*}
Q \equiv \frac{1}{n} X^\intercal x=\frac{1}{n} \sum_{i=1}^nX_ix_i^\intercal\,, \qquad \text{or componentwise} \qquad Q_{kk'}\equiv \frac{1}{n}	\sum_{i=1}^nX_{ik}x_{ik'}\,.
\end{align*}
Here $x$ is a sample drawn according to the posterior distribution and $X$ is the signal (all vectors are columns, including isolated rows of matrices, and transposed vectors are rows). The overlap contains a lot of information. E.g., the minimum mean-square error (MMSE), an error metric often considered in signal processing, is related to it through
\begin{align}
{\rm MMSE}\equiv\min_{\widehat x}\frac1n\EE\big[\|X-\widehat x(\widetilde Y,\theta)\|^2_{\rm F}\big]	 =\frac1n\EE\big[\|X- \langle x \rangle_0\|^2_{\rm F}\big] = \EE\big[\|X_1\|^2\big] - {\rm Tr}\,\EE\langle Q\rangle_0 \label{scalMMSE}
\end{align}
where we denote $\langle - \rangle_0$ the expectation w.r.t. the posterior \eqref{post} of the base model. The minimization is over all functions of $(\widetilde Y,\theta)$ in $\mathbb{R}^{n\times K}$, $\| -\|_{\rm F}$ is the Frobenius norm, $X_i=(X_{ik})_k\in [-S,S]^K$ is the $i$-th row of $X$. A simple fact from Bayesian inference is that the estimator minimizing the MMSE is the posterior mean $\langle x \rangle_0\equiv \EE[X|\widetilde Y,\theta]$. One may also be interested in the $K\times K$ MMSE matrix, which provides information about the individual dimensions in the row space of $X$:
\begin{align*}
\frac1n\EE\big[(X- \langle x \rangle_0)^\intercal(X- \langle x \rangle_0) \big] = \EE\big[X_1^\intercal X_1\big] - \EE\langle Q\rangle_0\,.
\end{align*}
This can be important in settings where some dimensions can be recovered while others cannot (see \cite{reeves2018mutual} and references therein) or, e.g., to study the ``specialization phase transition'' of the neurons during learning in some models of neural networks \cite{SM_Arxiv_Aubin2018}. The usual scalar MMSE \eqref{scalMMSE} is just the trace of this richer object.

Another metric of interest in problems where, e.g., the sign of the signal is lost due to symmetries is the matrix-MMSE (not to be confused with the MMSE matrix above). Again, it is related to the overlap (the notation $A=B+{\cal O}_S(1/n)$ means $|A-B|\le C(S)/n$ for some positive constant $C(S)$ depending only on the prior support $S$):
\begin{align}\label{m-MMSE}
{\rm mMMSE}\equiv\frac{1}{n^2}\sum_{i,j=1}^n\EE\big[(X_i^\intercal X_j- \langle x_i^\intercal x_j\rangle_0)^2\big] = \EE\big[(X_1^\intercal X_2)^2\big] - \EE\big\langle \|Q\|^2_{\rm F}\big\rangle_0 + {\cal O}_S(1/n)\,.
\end{align}
Finally if one is interested in estimating the sum over a subset ${\cal S}\subseteq \{1,\ldots,K\}$ of the the signal entries a possible error metric is
\begin{align*}
\frac{1}{n}\sum_{i=1}^n\EE\Big[\Big(\sum_{k\in{\cal S}}X_{ik}- \Big\langle \sum_{k\in{\cal S}} x_{ik}\Big\rangle_0\Big)^2\Big] = \sum_{(k,k')\in{\cal S}^2}\big(\EE[X_{1k}X_{1k'}] - \EE\langle Q_{kk'}\rangle_0\big)\,.	
\end{align*}
\subsection{Examples}
Let us provide some examples of models that fall under the setting of optimal Bayesian inference with vector variables as described in the previous section. 

In the symmetric order-$p$ rank-$K$ tensor factorization problem, the data-tensor $\widetilde Y=(\widetilde Y_{i_1\ldots i_p})$ is generated through the observation model 
\begin{align}
\widetilde Y_{i_1\ldots i_p}=n^{\frac{1-p}{2}}\,\sum_{k=1}^K X_{i_1 k}X_{i_2 k}\ldots X_{i_p k} + \widetilde Z_{i_1\ldots i_p}\,, \qquad
1\le i_1\le i_2\le \ldots \le i_p\le n\,.\label{ex:tensorFacto}
\end{align}
Here $\widetilde Z$ is a Gaussian noise tensor with independent and identically distributed (i.i.d.) ${\cal N}(0,1)$ entries for $1\le i_1\le i_2\le \ldots \le i_p\le n$, and the signal components are i.i.d., i.e., with a prior of the form $P_0=p_0^{\otimes n}$ with $p_0$ a probability distribution supported on $[-S,S]^K$. The case $p=2$ is known as the Wigner spike model, or low-rank matrix factorization, and is one of the simplest probabilistic model for principal component analysis. In both the analysis of \cite{Luneau_matFacto,mourrat2019hamilton} the matrix overlap concentration is a key result. The Wigner spike model is an example of model where the signal's sign is lost, and therefore a relevant error metric is the matrix-MMSE \eqref{m-MMSE}.

Another model is the following generalized linear model (GLM) (recall $X_i\in [-S,S]^K$):
\begin{align}
\widetilde Y_\mu \sim p_{\rm out}\Big(\,\cdot\, \Big| \sum_{i=1}^n \theta_{\mu i}X_{i}\Big), \qquad 1\le\mu\le m\,.\label{ex:GLM}
\end{align}
Note that here the $m$ observations are i.i.d. given $\mathbb{R}^{m\times n}\ni\theta_{\rm out}=(\theta_\mu)_{\mu=1}^m$ and $X$; this is the reason for the notation $p_{\rm out}$ instead of $P_{\rm out}$, the latter representing the full likelihood while the former is the conditional distribution of a single data point, i.e., $P_{\rm out}(\,\cdot\,|\theta_{\rm out} X)=\otimes_{\mu=1}^m p_{\rm out}(\,\cdot\,| X^\intercal \theta_\mu)$. We also assume that the prior $P_0=p_0^{\otimes n}$ is decoupled over the $n$ signal components and $m =\Theta(n)$. A particular simple deterministic case is 
\begin{align}
\widetilde Y_\mu ={\sign}\sum_{k=1}^K {\sign}\sum_{i=1}^n \theta_{\mu i}X_{ik}\,, \qquad 1\le\mu\le m\,. \label{committee}
\end{align}
This model is the committee machine mentionned in the introduction \cite{barbier2017phase,SM_Arxiv_Aubin2018}. Here $(X_{ik})_{i=1}^n$ can be interpreted as the weights of the $k$-th hidden neuron, and $(\theta_\mu)$ are $n$-dimensional data points used to generate the labels $(\widetilde Y_\mu)$. The teacher-student scenario in which our results apply corresponds to the following: the teacher network \eqref{committee} (or \eqref{ex:GLM} in general) generates $\widetilde Y$ from the data $\theta_{\rm out}$. The pairs $(\widetilde Y_{\mu}, \theta_\mu)$ are then used in order to train (i.e., learn the weights of) a student network with exactly the same architecture.

A richer example is a multi-layer version of the GLM above:
\begin{align}\label{ex:multiLayer}
\begin{cases}
X^{(L)}_{i_L} \sim p_{\rm out}^{(L)}\big(\,\cdot\, \big| \sum_{j=1}^{n_{L-1}} \theta_{i_L j}^{(L)}X_{j}^{(L-1)}\big)\,, &1\le i_L\le n_L\,,\\
X_{i_{L-1}}^{(L-1)}\sim  p_{\rm out}^{(L-1)}\big(\,\cdot\, \big| \sum_{j=1}^{n_{L-2}} \theta_{i_{L-1} j}^{(L-1)}X_{j}^{(L-2)}\big)\,, &1\le i_{L-1}\le n_{L-1}\,,\\
\qquad\qquad\qquad\qquad\vdots\\
X_{i_1}^{(1)}\sim  p_{\rm out}^{(1)}\big(\,\cdot \,\big| \sum_{j=1}^{n_0} \theta_{i_1 j}^{(1)}X_{j}^{(0)}\big)\,,  &1\le i_1\le n_1\,.	
\end{cases}
\end{align}
with an input $X^{(0)}\sim P_{0}$ factorized as $P_0=p_0^{\otimes n_0}$. In this model $(X^{(\ell)})_{\ell=1}^{L-1}$ represent intermediate hidden variables, the visible variable $X^{(L)}=\widetilde Y$ is the data, and $\theta_{\rm out}=(\theta^{(\ell)})$ with $\theta^{(\ell)}$ representing the weight matrix at the $\ell$-th layer. Note that in the single layer version \eqref{ex:GLM}, $\theta_{\rm out}$ was instead interpreted as data points and $X$ was the weight vector to learn/infer. Also $n^{(\ell)}=\Theta(n_0)$ for $\ell=1,\ldots,L$. This scaling for the variables sizes is often assumed in order not to make the inference of $X^{(0)}$ from $X^{(L)}$ impossible, nor trivial. This multi-layer GLM has been studied by various authors for the $K=1$ case and when the output components $X_{j}^{(\ell)}$ are scalars \cite{manoel_multi-layer_2017,reeves_additivity_2017,fletcher_inference_2017,Gabrie:NIPS2018,DBLP:journals/corr/abs-1903-01293}. But one can define generalizations where these are multi-dimensional, in which case overlap matrices naturally arise.

A final example could be another combination of complex statistical models such as, e.g., the following symmetric matrix factorization problem where the hidden low-rank representation $X$ of the matrix is itself generated from a generalized linear model over a more primitive signal $X^{(0)}$:
\begin{align}\label{ex:mix}
\begin{cases}
\widetilde Y_{ij}=n^{-1/2}\,\sum_{k=1}^K X_{ik}X_{jk}+\widetilde Z_{ij}\,,  &1\le i\le j\le n\,,\\
X_i \sim p_{\rm out}\big(\,\cdot\, \big| \sum_{j=1}^{n_0} \theta_{i j}X_{j}^{(0)}\big)\,, &1\le i\le n\,.	
\end{cases}
\end{align}
Here again some factorization structure for the prior $P_0=p_0^{\otimes n_0}$ of $X^{(0)}$ may be assumed, and $n=\Theta(n_0)$. Such model has recently been studied in \cite{aubin2019spiked}.
\subsection{The Nishimori identity}
The following identity is a simple consequence of Bayes' formula, and applies to optimal Bayesian inference.
\begin{lemma}[Nishimori identity]\label{NishId}
Let $(X,Y)$ be a couple of random variables with joint distribution $P(X, Y)$ and conditional distribution 
$P(X | Y)$. Let $k \geq 1$ and let $x^{(1)}, \dots, x^{(k)}$ be i.i.d.\ samples from the conditional distribution $P(\,\cdot\,|Y)$. These are called ``replicas''. Let us denote $\langle - \rangle$ the expectation operator w.r.t. the product conditional distribution $P(x^{(1)}|Y)P(x^{(2)}|Y)\dots P(x^{(k)}|Y)$ acting on the replicas, and $\mathbb{E}$ the expectation w.r.t. the joint distribution $P(X,Y)$. Then, for any continuous bounded function $g$,
\begin{align*}
\mathbb{E} \big\langle g(Y,x^{(1)}, \dots, x^{(k)}) \big\rangle
=
\mathbb{E} \big\langle g(Y, X, x^{(2)}, \dots, x^{(k)}) \big\rangle\,. 
\end{align*}	
\end{lemma}
\begin{proof}
It is equivalent to sample the couple $(X,Y)$ according to its joint distribution or to sample first $Y$ according to its marginal distribution and then to sample $X$ conditionally on $Y$ from the conditional distribution. Thus the two $(k+1)$-tuples $(Y,x^{(1)}, \dots,x^{(k)})$ and $(Y, X, x^{(2)},\dots,x^{(k)})$ have the same law.	
\end{proof}

In practice the Nishimori identity\footnote{This identity has been abusively called ``Nishimori identity'' in the statistical physics literature despite that
it is a simple consequence of Bayes' formula. The 
``true'' Nishimori identity concerns models with one extra feature, namely a 
gauge symmetry which allows to eliminate the input signal, and the expectation over the signal $X$ in expressions of the form $\mathbb{E}\langle -\rangle$
can therefore be dropped.} allows to ``replace'' the ground-truth signal $X$ by an independent replica, and vice-versa, in expressions involving only other replicas and the observations. Again, by replicas we mean conditionally independent samples drawn according to the posterior. 

\section{The vectorial Gaussian channel perturbation}\label{sec:Gauss_channel}
In order to ``force'' the overlap to concentrate we need to have access to infinitesimal side-information $Y$ in addition to the observations $\widetilde Y$. This side information is coming from the following vectorial Gaussian channel:
\begin{align}
Y = X \lambda_n^{1/2} + Z\,, \qquad \text{or componentwise} \qquad Y_i = \lambda_n^{1/2}\, X_i +  Z_i \quad \text{for}\quad 1\le i\le n\,.\label{pert_channel}
\end{align}
Here the signal $X$ is the same as in the base inference model. The observations $(Y_i)$, the signal components $(X_i)$ and i.i.d. Gaussian noise variables $(Z_i)\sim {\cal N}(0,I_K)^{\otimes n}$ are all $K$-dimensional vectors. The signal-to-noise (SNR) matrix controlling the signal strength $$\lambda_n\equiv s_n \tilde \lambda\,,$$ with a positive sequence $(s_n)\in(0,1]^{\mathbb{N}}$ that tends to $0_+$ slowly enough (the rate will be specified later), and $\tilde\lambda$ belongs to ${\cal D}_{K}$ defined as
%
\begin{align*}
	{\cal D}_{K} \equiv \big\{\tilde\lambda\in\mathbb{R}^{K\times K} : \tilde\lambda_{kk'}=\tilde\lambda_{k'k}\in(1,2) \ \forall \ k\neq k', \tilde\lambda_{kk}\in(2K,2K+1) \ \forall\ k  \big\}\,.
\end{align*}
Therefore, $\lambda_n$ belongs to the set
\begin{align}\label{DnK}
	{\cal D}_{n,K} \equiv \big\{\lambda\in\mathbb{R}^{K\times K} : \lambda_{kk'}=\lambda_{k'k}\in(s_n,2s_n) \ \forall \ k\neq k', \lambda_{kk}\in(2Ks_n,(2K+1)s_n) \ \forall\ k  \big\}\,.
\end{align}
 Matrices belonging to ${\cal D}_{n,K}$ are symmetric strictly diagonally dominant with positive entries (of the order of $s_n$) and thus ${\cal D}_{n,K}\subset {\cal S}_{K}^+$, where ${\cal S}_{K}^+$ is the set of symmetric positive definite matrices of dimension $K \times K$, see \cite{horn1990matrix}. As $\lambda_n\in {\cal D}_{n,K}$ it possesses a unique principal square root matrix denoted $$\lambda_n^{1/2}=\sqrt{s_n}\,\tilde\lambda^{1/2}\,.$$ The advantage of working with the ensemble ${\cal D}_{n,K}$ is the following. We require that the SNR matrix $\lambda_n$ always belong to ${\cal S}_{K}^+$ so that its square root is real and unique. For a generic positive matrix in ${\cal S}_{K}^+$, but not necessarily in ${\cal D}_{n,K}$, one cannot vary its (symmetric) elements \emph{independently} because doing so the matrix might not be positive definite anymore; the constraint $\lambda_n\in{\cal S}_K^+$ is a ``global'' constraint over the matrix elements. In contrast if $\lambda_n \in {\cal D}_{n,K}$ we can vary its elements independently (as long as it remains in ${\cal D}_{n,K}$) without the possibility that $\lambda_n$ falls out of ${\cal S}_{K}^+$. 

The perturbed inference model is then the following obervation model:
\begin{align}\label{2channels}
	\begin{cases}
	\widetilde Y  \!\!\!\!\! &\sim P_{\rm out}(\,\cdot\, |X,\theta_{\rm out})\,,\\
	Y_i \!\!\!\!\!&=  \lambda_n^{1/2}\,X_i +  Z_i\,, \quad 1\le i\le n\,.
	\end{cases}
\end{align}
It is called ``perturbed model'' because the original observation model has been slightly modified by adding new observations coming from \eqref{pert_channel} that are ``weak'' (as $s_n\to 0_+$). The perturbation Hamiltonian associated with the observation channel \eqref{pert_channel} is
\begin{align}\label{added-pert}
{\cal H}_{\lambda}(x,Y(X,Z),\lambda_n)\equiv\sum_{i=1}^{n} \Big( \frac {1}{2}x_i^\intercal\lambda_n x_i - x_{i}^\intercal \lambda_n X_{i}-x_{i}^\intercal \lambda_n^{1/2} \, Z_{i}\Big)
\end{align} 
using the symmetry of the SNR matrix. The total Hamiltonian is therefore the sum of the base Hamiltonian and the perturbation one. The posterior of the perturbed model, written in the standard Gibbs-Boltzmann form of statistical mechanics, is
\begin{align}
P(x|\widetilde Y,Y,\theta,\lambda_n) = \frac{1}{{\cal Z}_n(\widetilde Y,Y,\theta,\lambda_n)} P_0(x|\theta_0)\exp\{-{\cal H}_0(x,\widetilde Y,\theta_{\rm out})-{\cal H}_{\lambda}(x,Y,\lambda_n)\} \label{post_pert}
\end{align}
where again the partition function ${\cal Z}_n(\widetilde Y,Y,\theta,\lambda_n)$ is simply the normalization constant. We also define the {\it Gibbs-bracket} $\langle -\rangle$ as the expectation operator w.r.t. the posterior of the perturbed model: 
\begin{align}
	\langle g\rangle \equiv \int dP(x|\widetilde Y,Y,\theta,\lambda_n)\, g(x) \label{bracket}
\end{align}
for any function $g$ s.t. its expectation exists. Thus $\langle g\rangle$ depends on the quenched variables $(\widetilde Y,Y,\theta)$ and the perturbation parameter $\lambda_n$.

It is crucial to notice the following. The perturbation is constructed from an inference channel \eqref{pert_channel} which form is known (i.e., it is known that the noise is a realization of ${\cal N}(0,I_n)$ and the signal-to-noise ratio matrix $\lambda_n$ is given). Therefore the perturbed model \eqref{2channels} is a proper inference problem in the optimal Bayesian inference setting. This means that in addition to the data $(\widetilde Y,Y)$, the statistician fully knows the data generating model, namely the likelihood $P_{\rm out}$ and the additive Gaussian nature of the noise in the second channel in \eqref{2channels}, the prior $P_0$ as well as all hyper-parameters $(\theta,\lambda_n)$, and is therefore able to write the true posterior \eqref{post_pert} of the model when estimating the signal. As a consequence {\it the Nishimori identity Lemma~\ref{NishId} applies to the perturbed model and its bracket $\langle -\rangle$}.	

An important quantity is the averaged {\it free energy} of the perturbed model:
\begin{align}
f_n=f_{n}(\lambda_n)&\equiv-\frac{1}{n}\EE \ln {\cal Z}_n(\widetilde Y,Y,\theta,\lambda_n)\nn
&=-\frac{1}{n}\EE \ln \int dP_0(x|\theta_0)\exp\{-{\cal H}_0(x,\widetilde Y,\theta_{\rm out})-{\cal H}_{\lambda}(x,Y,\lambda_n)\}\label{av_f}
\end{align}
where the above expectation $\mathbb{E}\equiv \mathbb{E}_\theta\mathbb{E}_{X|\theta_0}\mathbb{E}_{\widetilde{Y}|X,\theta_{\rm out}}\mathbb{E}_{Y|X,\lambda_n}$ carries over the random hyper-parameters, the ground-truth signal (given $\theta_0$) and the data generated according to \eqref{2channels}, but not over $\lambda_n$ which remains fixed. Later we will average quantities w.r.t. $\lambda_n\in{\cal D}_{n,K}$, but in this case we will explicitely write $\EE_\lambda$. 

In order to prove the concentration of the overlap we need the following {\it crucial} hypothesis:
\begin{hypothesis}[Free energy concentration]
	The free energy \eqref{av_f} of the perturbed model concentrates at the optimal rate, namely there exists a constant $C_f=C_f(K,P_0,P_{\rm out},P_{\theta_0},P_{\theta_{\rm out}})$ that may depend on everything but $n$, and s.t.
	\begin{align}
	\EE\Big[\Big(-\frac{1}{n} \ln {\cal Z}_n(\widetilde Y,Y,\theta,\lambda_n)-f_n(\lambda_n)\Big)^2\Big]\le \frac{C_f}{n}\,. \label{hyp:f_conc}
	\end{align}
\end{hypothesis} 	

There are some remarks to be made here. The first one is related to the scenarios where this hypothesis can be verified. For purely generic optimal inference models without any restricting assumptions on the form of the distributions $(P_0,P_{\rm out},P_{\theta_0},P_{\theta_{\rm out}})$ it is generally very hard, if not wrong, to try proving \eqref{hyp:f_conc}. The model must be ``random enough'' and possess some underlying factorization structure for such hypothesis to be true. The most studied case in the literature is when the prior and the likelihood factorize, namely $P_0=p_0^{\otimes n}$ and the data points are i.i.d. given $(X,\theta_{\rm out})$. The examples \eqref{ex:tensorFacto}--\eqref{committee} fall in this class. Under such independence/factorization assumptions it is quite straightforward to prove that the free energy concentrates using standard techniques\footnote{In the context of statistical physics of spin glasses, these factorization and independence properties translate into the fact that the external fields $(h_{i})$ act independently on each spins, and the coupling constants $(J_{ij})$ act pairwise (or on a finite subsets of variables for $p$-spin types of models), and are independently drawn from some distribution. Therefore, in this context, free energy concentration is standard to prove. This is not necessarily the case in inference with generic prior and noise models, that can induce correlations preventing self-averaging. This is the reason why free energy concentration is here stated as an hypothesis.} (see, e.g., \cite{BarbierM17a,barbier2017phase}). But such simple factorization properties are not always there, as illustrated by examples \eqref{ex:multiLayer}, \eqref{ex:mix}. In these two last examples it is a perfectly valid question to wonder whether the overlap of the hidden variables do concentrate\footnote{Note that proving concentration of the overlap for a hidden variable requires a perturbation of the form \eqref{pert_channel} over the hidden variable, not over $X^{(0)}$, which in this case is just interpreted as a constitutive element of the prior of the hidden variable of interest, see \cite{Gabrie:NIPS2018} where this is done.} (this question is crucial in the analysis of \cite{Gabrie:NIPS2018}). The hidden variables have very complex structured prior (i.e., probability distribution), with highly non-trivial factorization properties, in which case proving \eqref{hyp:f_conc} requires work. See, e.g., \cite{Gabrie:NIPS2018} where this has been done for the multi-layer GLM \eqref{ex:multiLayer} with a single hidden layer ($L=2$) where this is already challenging. 

The second remark is that the perturbation does not change the limit of the averaged free energy:

\begin{lemma}[The base and perturbed models have same asymptotic averaged free energy]\label{lemma:same_f} We have $$|f_{0,n} - f_{n} |	\le S^2(2K+1)K^2s_n\,.$$ Therefore $f_{0,n}$ and $f_{n}$ have same thermodynamic limit, provided it exists.\\

\begin{proof}
It follows from identities \eqref{80}, \eqref{80_ll} in section \ref{sec:4.1} that $\|\nabla_{\lambda_n} f_{n}\|_{\rm F} \le \|\EE\langle Q\rangle\|_{\rm F}$. By the mean value theorem $|f_{0,n} - f_{n} |\le \|\nabla_{\lambda_n} f_{n}\|_{\rm F} \| \lambda_n\|_{\rm F}$ and thus $|f_{0,n} - f_{n} |\le  \|\EE\langle Q\rangle\|_{\rm F} \| \lambda_n\|_{\rm F}$. By definition matrices in ${\cal D}_{n,K}$ have positive entries bounded by $(2K+1)s_n$. Therefore as $\lambda_n\in{\cal D}_{n,K}$ so $\| \lambda_n\|_{\rm F}\le (2K+1)Ks_n$. By hypothesis the prior support is contained in $[-S,S]$ so we have $\|\EE\langle Q\rangle\|_{\rm F}\le S^2K$, and thus the result.
\end{proof}
\end{lemma}
\subsection{Main results}
All along this paper we denote $C(U)$ a generic positive numerical constant depending only on the parameters $U$. E.g., $C(C_f,K,S)>0$ depends only on $C_f$ appearing in \eqref{hyp:f_conc}, the variables dimensionality $K$ and on the prior support $S$. Let us denote the average over the matrix $\lambda_n\in {\cal D}_{n,K}$ appearing in the perturbation \eqref{pert_channel} as $$\EE_{\lambda}[-]\equiv \frac{1}{{\rm Vol}({\cal D}_{n,K})}\int_{{\cal D}_{n,K}} d\lambda_n\,[-]\,, \quad \text{with}\quad  {\rm Vol}({\cal D}_{n,K})\equiv\int_{{\cal D}_{n,K}} d\lambda_n=s_n^{K(K+1)/2}\,.$$
Here ${\rm Vol}({\cal D}_{n,K})$ is the volume of ${\cal D}_{n,K}$ which vanishes as $n\to +\infty$ (there are $K(K+1)/2$ independent entries in $\lambda_n$ as it is symmetric). Recall the notation $\langle -\rangle$ for the expectation  w.r.t. the posterior of the perturbed inference model \eqref{bracket}.

In order to give our first result we need to introduce the overlap between two replicas 
\begin{align}\label{Q12}
Q^{(12)}\equiv\frac1n \sum_{i=1}^n x^{(1)}_i(x_i^{(2)})^\intercal\,,	
\end{align}
where, again, replicas are conditionally i.i.d. random variables drawn accroding to the posterior \eqref{post_pert} of the perturbed model (and thus share the same quenched variables): $(x^{(1)},x^{(2)})\sim P(\,\cdot\,|\widetilde Y,Y,\theta,\lambda_n)^{\otimes 2}$. By a slight abuse of notation let us continue to use the same bracket notation for the expectation of functions of replicas w.r.t. to the product posterior measure: $$\big\langle g(x^{(1)},x^{(2)})\big\rangle\equiv\int dP(x^{(1)}|\widetilde Y,Y,\theta,\lambda_n)dP(x^{(2)}|\widetilde Y,Y,\theta,\lambda_n) g(x^{(1)},x^{(2)})\,.$$ 

Our main results are the following concentration theorems for the overlap in a (perturbed) model of optimal Bayesian inference. We start with the first type of fluctuations, namely the fluctuations of the overlap w.r.t. the posterior distribution, or what is called ``thermal fluctuations'' in statistical mechanics. Note that for controlling these fluctuations we do not need that the free energy concentrates, i.e., the hypothesis \eqref{hyp:f_conc} is not required. As a consequence this result is valid even for very complex models without any factorization properties for the signal's prior nor for the likelihood (as long as they are defined in the optimal Bayesian setting). This result is a consequence of the precense of the perturbation combined with the Nishimori identity.
\begin{theorem}[Thermal fluctuations of $Q$]\label{prop:thermalBOund}
Assume that the perturbed inference model is s.t. the Nishimori identity Lemma~\ref{NishId} holds. Let $(s_n)\in(0,1]^{\mathbb{N}}$ a positive sequence verifying $s_n\to 0_+$ and $s_nn\to+\infty$. There exists positive constants $C(K,S)$ s.t.
\begin{align}
\EE_\lambda\EE\big\langle \|Q- \langle Q\rangle\|_{\rm F}^2 \big\rangle &\le \frac{C(K,S)}{\sqrt{s_nn}} \,,\label{36}\\
\EE_\lambda\EE\big\langle \|Q-\langle Q^{(12)}\rangle\|_{\rm F}^2	\big\rangle   &\le \frac{C(K,S)}{\sqrt{s_nn}} \,.\label{2terms_controlled}
\end{align}
\end{theorem}

The next, stronger, result takes care of the additional fluctuations due to the quenched randomness, and requires the free energy concentration hypothesis:
\begin{theorem}[Total fluctuations of $Q$]\label{thm:Q_con}
Assume that the perturbed inference model is s.t. $i)$ its free energy concentrates as in identity \eqref{hyp:f_conc}; $ii)$ the Nishimori identity Lemma~\ref{NishId} holds. Let $(s_n)\in(0,1]^{\mathbb{N}}$ a positive sequence verifying $s_n\to 0_+$ and $s_n^4n\to+\infty$. There exists a positive constant $C(C_f,K,S)$ s.t.
\begin{align*}
\EE_\lambda\EE\big\langle \|Q-\EE\langle Q\rangle\|_{\rm F}^2\big\rangle \le \frac{C(C_f,K,S)}{(s_n^4n)^{1/6}}\,.
\end{align*}	
\end{theorem}

Before entering the proof let us make a very last remark. There are problems with multiple overlaps. For example one may also consider the non-symmetric version of the tensor factorization problem. In this case $p$ matrices $X^{[p]}\in \mathbb{R}^{n_p\times K}$, with $n_p=\Theta(n)$ and with a possibly matrix-dependent prior $P_0^{[p]}$, are to be reconstructed from a data-tensor of the form
\begin{align*}
\widetilde Y_{i_1\ldots i_p}=n^{\frac{1-p}{2}}\,\sum_{k=1}^K X_{i_1 k}^{[1]}X_{i_2 k}^{[2]}\ldots X_{i_p k}^{[p]} + \widetilde Z_{i_1\ldots i_p}\,, \qquad 1\le i_1\le n_1, \, 1\le i_2\le n_2,\, \ldots, 1\le i_p\le n_p \,.
\end{align*}
In this case there is one overlap per matrix-signal to be inferred: $$Q^{[p']} \equiv \frac1{n_{p'}} \sum_{i=1}^{n_{p'}} X_{i}^{[p']}(x_i^{[p']})^\intercal\,, \qquad 1\le p'\le p\,.$$ It should be clear to the reader that all the setting described in this paper can be straightforwardly extended to include this case: one has to consider one perturbation channel of the form \eqref{pert_channel} per variable to be reconstructed (i.e., per matrix in the non-symmetric tensor factorization problem), each with its own independent matrix SNR: $$Y^{[p']} = X^{[p']}(\lambda_n^{[p']})^{1/2} +  Z^{[p']}\,, \qquad 1\le p'\le p\,.$$ Then the total Hamiltonian is the sum of the base one and the $p$ perturbation Hamiltonians, and so forth.
\section{Proof of concentration of the overlap matrix} \label{sec4}
For the sake of readibility we now drop the $n$ index in the matrix SNR: $$\lambda= \lambda_n\in {\cal D}_{K,n}\,.$$We use $l$, $l'$ and $k$, $k'$ for the variables dimension indices which are running from $1$ to $K$, and $i$, $j$ for the variables indices running from $1$ to $n$. When we write $l$ and $l'$ we always implicitly mean $l'\neq l$.

Let us start with some preliminary computations.
\subsection{Preliminaries: properties of the matrix ${\cal L}$}\label{sec:4.1}
The proof that the overlap concentrates relies on the concentration of another matrix defined as
\begin{align}
\mathcal{L}_{ll'} &\equiv \frac{1}{n}\frac{d{\cal H}_{\lambda}}{d\lambda_{ll'}}= \frac1{n}\sum_{i=1}^n\Big(x_{il} x_{il'} -x_{il} X_{il'}-
x_{il'} X_{il} -   x_{i}^\intercal \frac{d\lambda^{1/2}}{d\lambda_{ll'}}  Z_{i}\Big)\,,\label{66}\\
\mathcal{L}_{ll} &\equiv \frac{1}{n}\frac{d{\cal H}_{\lambda}}{d\lambda_{ll}}= \frac1{n}\sum_{i=1}^n\Big( \frac12x_{il}^2 -
x_{il} X_{il} -   x_{i}^\intercal \frac{d\lambda^{1/2}}{d\lambda_{ll}}  Z_{i}\Big)\,, \label{67}
\end{align}
where we used $$\Big(\frac{d\lambda}{d\lambda_{ll'}}\Big)_{kk'}=\delta_{kl}\delta_{k'l'}+\delta_{kl'}\delta_{k'l}\,,\qquad \Big(\frac{d\lambda}{d\lambda_{ll}}\Big)_{kk'}=\delta_{kl}\delta_{k'l}\,.$$
The fluctuations of this matrix are easier to control than the ones of the overlap. This comes from the fact that ${\cal L}$ is related to derivatives of the free energy, which is self-averaging by hypothesis \eqref{hyp:f_conc}. First consider $l'\neq l$. We have
\begin{align}\label{75}
\frac{df_n(\lambda)}{d\lambda_{ll'}}=\mathbb{E}\langle \mathcal{L}_{ll'}\rangle &= \frac1{n}\sum_{i=1}^n\mathbb{E}\Big\langle x_{il} x_{il'} - 
x_{il} X_{il'} - 
x_{il'} X_{il} -  x_i^\intercal \frac{d\lambda^{1/2}}{d\lambda_{ll'}}  Z_{i}\Big\rangle\nn
&\overset{\rm N}{=} \frac1n\sum_{i=1}^n\mathbb{E}\Big[ \langle x_{il} x_{il'}\rangle -2 
\langle x_{il}\rangle  \langle x_{il'} \rangle - \langle x_i\rangle^\intercal \frac{d\lambda^{1/2}}{d\lambda_{ll'}}  Z_{i}\Big]\,.
\end{align}
We used the Nishimori identity Lemma \ref{NishId} which in this case implies $$\EE\big[\langle x_{il} \rangle X_{il'}\big] = \EE\big[\langle x_{il}\rangle \langle x_{il'} \rangle\big]\,.$$ Each time we use an identity that is a consequence of Lemma \ref{NishId} we write a $\rm N$ on top of the equality (that stands for Nishimori). 
We integrate by part the Gaussian noise thanks to the formula $$\EE[Zg(Z)]=\EE \,g'(Z) \quad \text{for}  \quad Z\sim{\cal N}(0,1)$$ and any bounded function $g$. This leads to
\begin{align}
\EE\Big\langle x_i^\intercal \frac{d\lambda^{1/2}}{d\lambda_{ll'}}  Z_{i}\Big\rangle &=\sum_{k,k'=1}^K \EE\Big[\langle x_{ik}\rangle \Big(\frac{d\lambda^{1/2}}{d\lambda_{ll'}}\Big)_{kk'} Z_{ik'}\Big]= \sum_{k,k'=1}^K \EE\Big[\frac{d\langle  x_{ik}\rangle}{d Z_{ik'}} \Big(\frac{d\lambda^{1/2}}{d\lambda_{ll'}}\Big)_{kk'}\Big]\nn
&=\sum_{k,k'=1}^K \EE\Big[\Big(\langle x_{ik}(\lambda^{1/2}x_i)_{k'}\rangle - \langle x_{ik}\rangle (\lambda^{1/2}\langle x_i\rangle)_{k'} \Big) \Big(\frac{d\lambda^{1/2}}{d\lambda_{ll'}}\Big)_{kk'}\Big]\nn
&=\EE\Big[\Big\langle x_i^\intercal\frac{d\lambda^{1/2}}{d\lambda_{ll'}}  \lambda^{1/2} x_i\Big\rangle-\langle x_i\rangle^\intercal \frac{d\lambda^{1/2}}{d\lambda_{ll'}}\lambda^{1/2}\langle  x_i\rangle \Big]\,.
\label{76}
\end{align}
We used that the derivative of the Hamiltonian \eqref{added-pert} is
\begin{align}\label{77}
\frac{d{\cal H}_\lambda}{dZ_{ik}}=-(\lambda^{1/2}x_i)_k\,, \qquad \text{and thus} \qquad \frac{d\langle \cdot \rangle}{dZ_{ik}} = \langle \cdot \,(\lambda^{1/2}x_i)_k\rangle -  \langle \cdot \rangle (\lambda^{1/2}\langle x_i\rangle )_k\,.
\end{align}
We now exploit the symmetry of the matrices $x_ix_i^\intercal$ and $\langle x_i\rangle \langle x_i\rangle^\intercal$ in order to symmetrize the terms in \eqref{76} and then use the formula
\begin{align}\label{57}
 \lambda^{1/2}\frac{d\lambda^{1/2}}{d\lambda_{ll'}} + \frac{d\lambda^{1/2}}{d\lambda_{ll'}}\lambda^{1/2} = \frac{d\lambda}{d\lambda_{ll'}}\,.
\end{align}
Identity \eqref{76} then becomes
\begin{align}
\EE\Big\langle x_i^\intercal \frac{d\lambda^{1/2}}{d\lambda_{ll'}}  Z_{i}\Big\rangle &=\frac12 \EE\Big[\Big\langle x_i^\intercal\Big\{\frac{d\lambda^{1/2}}{d\lambda_{ll'}}  \lambda^{1/2}+\lambda^{1/2}\frac{d\lambda^{1/2}}{d\lambda_{ll'}}   \Big\}x_i\Big\rangle-\langle x_i\rangle ^\intercal\Big\{\frac{d\lambda^{1/2}}{d\lambda_{ll'}}  \lambda^{1/2}+\lambda^{1/2}\frac{d\lambda^{1/2}}{d\lambda_{ll'}}   \Big\}\langle x_i \rangle \Big]\nn
 &=\frac{1}{2}\EE\Big[\Big\langle x_i^\intercal\frac{d\lambda}{d\lambda_{ll'}}x_i\Big\rangle -\langle x_i\rangle^\intercal\frac{d\lambda}{d\lambda_{ll'}} \langle  x_i\rangle\Big]\nn
 &=\EE\big[\langle  x_{il}x_{il'}\rangle-\langle x_{il}\rangle\langle x_{il'}\rangle\big]\,.
\label{77_bis}
\end{align}
%
%
%
%
Similarly wo obtain for the diagonal terms $$\EE\Big\langle x_i^\intercal \frac{d\lambda^{1/2}}{d\lambda_{ll}}  Z_{i}\Big\rangle = \frac12\EE\big[\langle x_{il}^2\rangle-\langle x_{il}\rangle^2 \big]\,.$$
%
Using this \eqref{75} becomes 
\begin{align}
\frac{df_n(\lambda)}{d\lambda_{ll'}}=\mathbb{E}\langle \mathcal{L}_{ll'}\rangle	&=\mathbb{E}\langle \mathcal{L}_{l'l}\rangle=-\frac1{n}\EE\sum_{i=1}^n\langle x_{il}\rangle \langle x_{il'}\rangle\overset{\rm N}{=}-\frac1{n}\EE\sum_{i=1}^n X_{il}\langle x_{il'}\rangle = -\EE\langle Q_{ll'}\rangle=-\EE\langle Q_{l'l}\rangle\,,\label{80}\\
\frac{df_n(\lambda)}{d\lambda_{ll}}=\mathbb{E}\langle \mathcal{L}_{ll}\rangle	&=-\frac1{2n}\EE\sum_{i=1}^n\langle x_{il}\rangle^2\overset{\rm N}{=}-\frac1{2n}\EE\sum_{i=1}^n X_{il}\langle x_{il}\rangle = -\frac12\EE\langle Q_{ll}\rangle\,.\label{80_ll}
\end{align}
Therefore the expectation of ${\cal L}$ is directly related to the one of $Q$. It is thus natural to guess that if $\cal L$ concentrates onto its mean, the overlap should concentrate too. Indeed, the following concentration identity for ${\cal L}$ is key in proving Theorems~\ref{prop:thermalBOund} and \ref{thm:Q_con}. Note that the following proposition does not require the Nishimori identity (i.e., to be in the optimal Bayesian setting). But the Nishimori identity will be crucial when linking the fluctuations of $\cal L$ to those of $Q$.
\begin{proposition}[Concentration of ${\cal L}$]\label{prop:Lconc}
Let $(s_n)\in(0,1]^{\mathbb{N}}$ a positive sequence verifying $s_n\to 0_+$ and $s_nn\to+\infty$. There exists a positive constant $C(S,K)$ s.t.
\begin{align}
\EE_\lambda \mathbb{E} \big\langle \|\mathcal{L} - \langle \mathcal{L}\rangle\|_{\rm F}^2 \big\rangle  &\le \frac{C(S,K)}{s_nn}\,.\label{thermal_L}
\end{align}
Moreover if $s_n^4n\to+\infty$ and the free energy concentrates as in identity \eqref{hyp:f_conc}, then there exists a constant $C(C_f,K,S)$ s.t.
\begin{align}
\EE_\lambda \EE\big\langle \|{\cal L} -\EE\langle {\cal L}\rangle\|_{\rm F}^2 \big\rangle &\le  \frac{C(C_f,K,S)}{(s_n^{4}n)^{1/3}}\,.	\label{total_L}	
\end{align}
\end{proposition}
Let us assume this result at the moment and show how it implies concentration of $Q$. We will then prove Proposition~\ref{prop:Lconc} later in section~\ref{app:Lconc}.
\subsection{Thermal fluctuations: proof of Theorem~\ref{prop:thermalBOund}}
%

Let us start with the control of the fluctuations due to the posterior distribution. \\

\begin{proof}[Proof of \eqref{36} in Theorem~\ref{prop:thermalBOund}] 
By definition of the overlap we have
\begin{align}
\EE_\lambda\EE\big\langle (Q_{ll'}- \langle Q_{ll'}\rangle)^2 \big\rangle&=\EE_\lambda\EE\langle Q_{ll'}^2\rangle- \EE_\lambda\EE\big[\langle Q_{ll'}\rangle^2\big]\nn
&=\frac1{n^2}\sum_{i,j=1}^n\EE_\lambda\EE\big[ X_{il}X_{jl} (\langle x_{il'}x_{jl'}\rangle-\langle x_{il'}\rangle\langle x_{jl'}\rangle)\big]\label{toUse}\\
&\le\Big\{\frac1{n^2}\sum_{i,j=1}^n\EE\big[(X_{il}X_{jl})^2\big]\Big\}^{1/2} \Big\{\frac1{n^2}\sum_{i,j=1}^n\EE_\lambda\EE\big[(\langle x_{il'}x_{jl'}\rangle-\langle x_{il'}\rangle\langle x_{jl'}\rangle)^2\big] \Big\}^{1/2}\nonumber
\end{align}
using the Cauchy-Schwarz inequality. Note that the first term on the r.h.s. of this inequality is bounded by $C(S)$. We show next, using the Nishimori identity, that
\begin{align}
\frac{1}{2n^2}\sum_{i,j=1}^n \EE\big[(\langle x_{il}x_{jl}\rangle -\langle x_{il}\rangle \langle x_{jl}\rangle)^2\big]\le \mathbb{E}\big\langle (\mathcal{L}_{ll} - \langle \mathcal{L}_{ll}\rangle)^2\big\rangle +\frac{C(K,S)}{s_nn}\,\label{141}.
\end{align}
Thus we obtain, for large enough constants $C(K,S)$ and as $s_nn\to +\infty$,
\begin{align*}
\EE_\lambda\EE\big\langle \|Q- \langle Q\rangle\|_{\rm F}^2 \big\rangle \le C(K,S)\sum_{l}\Big\{\EE_\lambda\mathbb{E}\big\langle (\mathcal{L}_{ll} - \langle \mathcal{L}_{ll}\rangle)^2\big\rangle+\frac{C(K,S)}{s_nn}\Big\}^{1/2} \,.
\end{align*}
The concentration identity \eqref{thermal_L} in Proposition~\ref{prop:Lconc} then implies \eqref{36}.
%
%
%

It remains to prove the crucial identity \eqref{141}. Acting with the operator $\frac1n \frac{d}{d{\lambda_{ll}}}$ on both sides of \eqref{80_ll}, i.e., starting from the identity $$\frac1n \frac{d}{d{\lambda_{ll}}}\EE\langle {\cal L}_{ll}\rangle=-\frac{1}{2n} \frac{d}{d{\lambda_{ll}}}\EE\langle Q_{ll}\rangle$$ we obtain
\begin{align}
 &-\mathbb{E}\big\langle (\mathcal{L}_{ll} - \langle \mathcal{L}_{ll}\rangle)^2\big\rangle 
 + \frac{1}{n}\mathbb{E}\Big\langle \frac{d\mathcal{L}_{ll}}{d\lambda_{ll}} \Big\rangle=\frac1{2n}\sum_{i=1}^n\EE\big[X_{il}(\langle x_{il}{\cal L}_{ll} \rangle - \langle x_{il}\rangle \langle{\cal L}_{ll}\rangle)\big].
  \label{intermediate_}
\end{align}
Computing the derivative of $\mathcal{L}_{ll}$ and using $$-2\Big(\frac{d\lambda^{1/2}}{d\lambda_{ll}}\Big)^2=\lambda^{1/2}\frac{d^2\lambda^{1/2}}{d\lambda_{ll}^2}+\frac{d^2\lambda^{1/2}}{d\lambda_{ll}^2}\lambda^{1/2}$$ which follows from \eqref{57} we find, using \eqref{67} and similar computations as \eqref{76}--\eqref{77_bis}, that
\begin{align}
\Big|\frac{1}{n}\mathbb{E}\Big\langle \frac{d\mathcal{L}_{ll}}{d\lambda_{ll}} \Big\rangle\Big| &=\Big|\frac1{n^2}\sum_{i=1}^n\EE\Big[\Big\langle x_i^\intercal\Big(\frac{d\lambda^{1/2}}{d\lambda_{ll}}\Big)^2x_i\Big\rangle-\langle x_i\rangle^\intercal \Big(\frac{d\lambda^{1/2}}{d\lambda_{ll}}\Big)^2\langle x_i\rangle\Big]\Big|\nn
&\qquad\qquad\qquad\qquad\qquad\qquad\qquad\qquad\le \frac{2S^2}{n}\Big\|\Big(\frac{d\lambda^{1/2}}{d\lambda_{ll}}\Big)^2\Big\|_{\rm F} =  \frac{C(S,K)}{s_nn}\,.\label{bound_derL}
\end{align}
We used for the last step that the entries of the matrix $\big(\frac{d\lambda^{1/2}}{d\lambda_{ll}}\big)^2$ are ${\cal O}(s_n^{-1})$. Indeed recall that $$\lambda^{1/2}=\lambda_n^{1/2}\equiv \sqrt{s_n}\, \tilde\lambda^{1/2}$$ where $\tilde\lambda^{1/2}$ is independent of $n$. Therefore for any $(l,l')$
\begin{align}\label{38}
\frac{d\lambda^{1/2}}{d\lambda_{ll'}}= \sqrt{s_n} \,\frac{d{\tilde\lambda}^{1/2}}{d\tilde\lambda_{ll'}} \frac{d\tilde \lambda_{ll'}}{d\lambda_{ll'}}={\cal O}(s_n^{-1/2})	
\end{align}
where, by a slight abuse of notation, we mean here that each element of this matrix is ${\cal O}(s_n^{-1/2})$. Let us compute the following term appearing in \eqref{intermediate_}:
\begin{align}
\frac{1}{2n}\sum_{i=1}^n\EE\big[X_{il}(\langle x_{il}{\cal L}_{ll} \rangle - \langle x_{il}\rangle &\langle{\cal L}_{ll}\rangle)\big]=\frac{1}{2n^2}\sum_{i,j=1}^n \EE\Big[\frac12X_{il}\langle x_{il}x_{jl}^2 \rangle - X_{il}X_{jl}\langle x_{il}x_{jl}\rangle- X_{il} \Big\langle x_{il}  x_{j}^\intercal \frac{d\lambda^{1/2}}{d\lambda_{ll}}Z_{j} \Big\rangle\nn 
&- \frac12X_{il}\langle x_{il}\rangle \langle x_{jl}^2 \rangle + X_{il}X_{jl}\langle x_{il}\rangle \langle x_{jl}\rangle+ X_{il}\langle x_{il}\rangle \Big\langle {x}_{j}^\intercal \frac{d\lambda^{1/2}}{d\lambda_{ll}} Z_{j}\Big\rangle \Big]\,.\label{46_th}
\end{align}
We need to simplify $$T\equiv \EE\Big[X_{il}\langle x_{il}\rangle\Big\langle {x}_{j}^\intercal \frac{d\lambda^{1/2}}{d\lambda_{ll}} \, {Z}_{j}\Big\rangle - X_{il}\Big\langle x_{il}{x}_{j}^\intercal \frac{d\lambda^{1/2}}{d\lambda_{ll}}  {Z}_{j}\Big\rangle\Big]\,.$$ Using similar manipulations as for obtaining \eqref{76}--\eqref{77_bis}, i.e., by symmetrizing when possible in order to use \eqref{57}, we simplify the first term in $T$:
\begin{align*}
 &\EE\Big[X_{il}\langle x_{il}\rangle\Big\langle {x}_{j}^\intercal \frac{d\lambda^{1/2}}{d\lambda_{ll}}  {Z}_{j}\Big\rangle\Big]\nn
 &\qquad=\EE\Big[X_{il}\langle  x_{j}\rangle^\intercal\frac{d\lambda^{1/2}}{d\lambda_{ll}}\lambda^{1/2} \langle x_j x_{il}\rangle - 2X_{il}\langle x_{il}\rangle \langle x_j\rangle^\intercal\frac{d\lambda^{1/2}}{d\lambda_{ll}}\lambda^{1/2}\langle x_j\rangle+ X_{il}\langle x_{il}\rangle \Big\langle x_j^\intercal \frac{d\lambda^{1/2}}{d\lambda_{ll}} \lambda^{1/2}x_j\Big\rangle \Big]\nn
 &\qquad=\EE\Big[X_{il}\langle  x_{j}\rangle^\intercal\frac{d\lambda^{1/2}}{d\lambda_{ll}}\lambda^{1/2} \langle x_j x_{il}\rangle - X_{il}\langle x_{il}\rangle \langle x_{jl}\rangle^2+ \frac12X_{il}\langle x_{il}\rangle \langle x_{jl}^2\rangle \Big]
 \,.
\end{align*}
Similarly the second term in $T$ is
\begin{align*}
 -\EE\Big[X_{il}\Big\langle x_{il} {x}_{j}^\intercal \frac{d\lambda^{1/2}}{d\lambda_{ll}}  {Z}_{j}\Big\rangle\Big]&=-\EE\Big[X_{il}\Big\langle  x_{j}^\intercal\frac{d\lambda^{1/2}}{d\lambda_{ll}}\lambda^{1/2} x_j x_{il}\Big\rangle - X_{il}\langle x_{il} x_j^\intercal\rangle\frac{d\lambda^{1/2}}{d\lambda_{ll}}\lambda^{1/2}\langle x_j\rangle \Big]\nn
 &=-\EE\Big[\frac12X_{il}\langle  x_{jl}^2x_{il}\rangle - X_{il}\langle x_{il} x_j^\intercal\rangle\frac{d\lambda^{1/2}}{d\lambda_{ll}}\lambda^{1/2}\langle x_j\rangle \Big]\,.
\end{align*}
Therefore, using again \eqref{57}, $T$ is the equal to
\begin{align*}
 T&=\EE\Big[X_{il}\langle  x_{j}\rangle^\intercal\Big\{\frac{d\lambda^{1/2}}{d\lambda_{ll}}\lambda^{1/2}+\lambda^{1/2}\frac{d\lambda^{1/2}}{d\lambda_{ll}}\Big\} \langle x_j x_{il}\rangle - X_{il}\langle x_{il}\rangle \langle x_{jl}\rangle^2+ \frac12X_{il}\langle x_{il}\rangle \langle x_{jl}^2\rangle -\frac12X_{il}\langle  x_{jl}^2x_{il}\rangle \Big]\nn
 &=\EE\Big[X_{il}\langle  x_{jl}\rangle \langle x_{jl} x_{il}\rangle - X_{il}\langle x_{il}\rangle \langle x_{jl}\rangle^2+ \frac12X_{il}\langle x_{il}\rangle \langle x_{jl}^2\rangle -\frac12X_{il}\langle  x_{jl}^2x_{il}\rangle \Big]\,.
\end{align*}
Plugging this expression in \eqref{46_th} and then simplifying using the Nishimori identity we obtain
\begin{align*}
\frac{1}{2n}\sum_{i=1}^n\EE\big[X_{il}(\langle x_{il}{\cal L}_{ll} \rangle - \langle x_{il}\rangle \langle{\cal L}_{ll}\rangle)\big]=-\frac{1}{2n^2}\sum_{i,j=1}^n  \EE\big[(\langle x_{il}x_{jl}\rangle -\langle x_{il}\rangle\langle x_{jl}\rangle)^2\big]\,.
\end{align*}
Together with \eqref{bound_derL} and \eqref{intermediate_} it ends the proof of \eqref{141}, and therefore of \eqref{36} too. 
\end{proof}

\begin{proof}[Proof of \eqref{2terms_controlled} in Theorem~\ref{prop:thermalBOund}] We denote the overlap between the replica $x=x^{(1)}$ and the ground-truth signal $X=x^{(0)}$ equivalently as $$Q=Q^{(01)}\equiv\frac1n \sum_{i=1}^n X_i(x_i^{(1)})^\intercal\,,$$ and recall definition \eqref{Q12} for the overlap between two replicas $x^{(1)}$ and $x^{(2)}$. Recall also that we still use the bracket notation $\langle -\rangle$ for the expectation w.r.t. $P(x^{(1)}|\widetilde Y,Y,\theta,\lambda_n)P(x^{(2)}|\widetilde Y,Y,\theta,\lambda_n)$, the product posterior measure acting of the conditionally independent replicas. 

The proof relies on the following relation which is a simple consequence of the Nishimori identity:
\begin{align}
\EE\big\langle \|Q^{(01)}-\langle Q^{(12)}\rangle\|_{\rm F}^2	\big\rangle &\overset{\rm N}{=}\EE\big\langle \|Q^{(12)}-\langle Q^{(12)}\rangle\|_{\rm F}^2 \big\rangle \label{Leo_0}\\
&=\EE\big\langle \|Q^{(12)}\|_{\rm F}^2\big\rangle -\EE\big[\|\langle Q^{(12)}\rangle\|_{\rm F}^2\big] \nn
&\overset{\rm N}{=}\EE\big\langle \|Q^{(01)}\|_{\rm F}^2\big\rangle -\EE\big[\|\langle Q^{(12)}\rangle\|_{\rm F}^2\big] \nn
&\hspace{-1cm}=\big(\EE\big\langle \|Q^{(01)}\|_{\rm F}^2\big\rangle -\EE\big[\|\langle Q^{(01)}\rangle\|_{\rm F}^2\big]\big)  +\big(\EE\big[\|\langle Q^{(01)}\rangle\|_{\rm F}^2\big]-\EE\big[\|\langle Q^{(12)}\rangle\|_{\rm F}^2\big]\big)\,.\label{Leo}
\end{align}
The first fluctuations in \eqref{Leo}, once averaged over $\lambda$, are controlled by \eqref{36} that we have just proven. The second fluctuations are controlled as follows:
\begin{align*}
	\EE_\lambda\EE\big[\langle Q_{ll'}^{(01)}\rangle^2-\langle Q_{ll'}^{(12)}\rangle^2\big]&=\frac1{n^2}\sum_{i,j=1}^n\EE_\lambda\EE\big[ X_{il}\langle x_{il'}\rangle X_{jl} \langle x_{jl'}\rangle - \langle x_{il}\rangle\langle x_{il'}\rangle\langle x_{jl}\rangle\langle x_{jl'}\rangle\big]\nn
	&\overset{\rm N}{=}\frac1{n^2}\sum_{i,j=1}^n\EE_\lambda\EE\big[ \langle x_{il'}\rangle \langle x_{jl'}\rangle(\langle x_{il}x_{jl}\rangle - \langle x_{il}\rangle\langle x_{jl}\rangle)\big]\,.
\end{align*}
We recognize a similar form as \eqref{toUse}. The derivation is then the same as the one of \eqref{36} based on \eqref{141} and yields
\begin{align*}
\EE_\lambda\Big[\EE\big[\|\langle Q^{(01)}\rangle\|_{\rm F}^2\big]-\EE\big[\|\langle Q^{(12)}\rangle\|_{\rm F}^2\big]\Big]&\le \frac{C(K,S)}{\sqrt{s_nn}} \,.	
\end{align*}
This ends the proof of \eqref{2terms_controlled}, and thus of Proposition~\ref{prop:thermalBOund}.
\end{proof}
\subsection{Total fluctuations: proof of Theorem~\ref{thm:Q_con}}
Now that we control the thermal fluctuations we are in position to prove our second main Theorem~\ref{thm:Q_con}. It shows that if the free energy concentrates then the overlap not only concentrates w.r.t. the posterior distribution, but also w.r.t. the quenched variables. The spirit of the proof is similar to the derivation of the Ghirlanda-Guerra identities in the context of spin glasses\footnote{What we mean here is that, as for the proof of the Ghirlanda-Guerra identities, the present method is based on testing the overlap matrix $Q$ against the fluctuations of the derivative of the Hamiltonian (the ${\cal L}$ matrix given by \eqref{66}, \eqref{67}). In the context of inference, related to the Nishimori line in spin glasses \cite{nishimori2001statistical,contucci2009spin}, the derived identities really are a special case of the Ghirlanda-Guerra identities.} \cite{GhirlandaGuerra:1998}. Our goal here is to compute 
\begin{align}
\EE_\lambda{\rm Tr}\,\EE \big\langle Q({\cal L} -\EE\langle {\cal L}\rangle)\big \rangle &=\EE_\lambda{\rm Tr}\,\EE \langle Q{\cal L}\rangle + \EE_\lambda\big[\| \EE\langle Q\rangle\|_{\rm F}^2\big] - \frac12 \sum_{l=1}^K\EE_\lambda\big[\EE[\langle Q_{ll}\rangle]^2\big]\,, \label{103easy}
\end{align}
using that $$\EE\langle {\cal L}\rangle = \frac12 {\rm diag}(\EE\langle Q\rangle)-\EE\langle Q\rangle$$ because of \eqref{80}, \eqref{80_ll}, and that $\EE \langle Q\rangle$ is symmetric. The crux of the proof is that from the quantity $\EE_\lambda{\rm Tr}\,\EE \langle Q({\cal L} -\EE\langle {\cal L}\rangle)\rangle$ will appear the fluctuations of the overlap, and this quantity is small by Proposition~\ref{prop:Lconc} (with Cauchy-Schwarz).

It thus only remains to compute $\EE_\lambda{\rm Tr}\,\EE \langle Q{\cal L}\rangle$. Let us first consider the off-diagonal terms:
\begin{align}
\sum_{l\neq l'} \EE_\lambda\EE \langle Q_{ll'}{\cal L}_{l'l}\rangle &=\sum_{l\neq l'} \EE_\lambda\EE\Big[ \frac{1}{n}\sum_{i=1}^n \langle Q_{ll'} x_{il} x_{il'}\rangle -\langle Q_{ll'}Q_{l'l}\rangle -\langle Q_{ll'}^2\rangle -   \frac{1}{n}\sum_{i=1}^n\Big\langle Q_{ll'} x_{i}^\intercal \frac{d\lambda^{1/2}}{d\lambda_{ll'}} \, Z_{i}\Big\rangle\Big]\,.\label{32}
\end{align}
We need to simplify the last term. Using \eqref{77},
\begin{align}
\EE_\lambda\EE\Big\langle Q_{ll'} x_{i}^\intercal \frac{d\lambda^{1/2}}{d\lambda_{ll'}} \, Z_{i}\Big\rangle &=	\EE_\lambda\EE\Big[\Big\langle Q_{ll'}x_i^\intercal\frac{d\lambda^{1/2}}{d\lambda_{ll'}} \lambda^{1/2} x_i  \Big\rangle - \langle Q_{ll'}x_i^\intercal\rangle\frac{d\lambda^{1/2}}{d\lambda_{ll'}} \lambda^{1/2} \langle x_i\rangle \Big]\nn
&=\EE_\lambda\EE\Big[\langle Q_{ll'}x_{il}x_{il'}  \rangle - \langle Q_{ll'}x_i^\intercal\rangle\frac{d\lambda^{1/2}}{d\lambda_{ll'}} \lambda^{1/2} \langle x_i\rangle \Big]\,.\label{33}
\end{align}
The first term of the r.h.s. of \eqref{33} has been simplified because $x_ix_i^\intercal$ is symmetric, allowing the symmetrization of $\frac{d\lambda^{1/2}}{d\lambda_{ll'}}$ followed by the use of \eqref{57}. In contrast the second term above lacks symmetry as the matrix $\langle x_i\rangle \langle Q_{ll'} x_i^\intercal\rangle$ is not symmetric. This prevents the use of the mechanism employed for the first term. In order to face this difficulty we exploit the concentration of the overlap w.r.t. the posterior that has been shown previously. We can write
\begin{align*}
\Big|\EE_\lambda\EE\Big[\langle Q_{ll'}x_i^\intercal\rangle\frac{d\lambda^{1/2}}{d\lambda_{ll'}} \lambda^{1/2} \langle x_i\rangle \Big]	- \EE_\lambda\EE\Big[\langle Q_{ll'}\rangle \langle x_i\rangle^\intercal\frac{d\lambda^{1/2}}{d\lambda_{ll'}} \lambda^{1/2} \langle x_i\rangle \Big]\Big| \le  \frac{C(K,S)}{(s_nn)^{1/4}}
\end{align*}
by relation \eqref{36} in Proposition~\ref{prop:thermalBOund} (which relies on the Nishimori identity) and Cauchy-Schwarz, because the entries of the matrix $\frac{d\lambda^{1/2}}{d\lambda_{ll'}} \lambda^{1/2}$ are bounded (recall that $\lambda^{1/2}\equiv \sqrt{s_n}\, \tilde\lambda^{1/2}$ where $\tilde\lambda^{1/2}$ is independent of $n$ with bounded entries, and \eqref{38}) as well as the support of the prior (so $|\langle x_{ik}\rangle|\le S$). Now we can exploit symmetry and therefore write
\begin{align*}
\EE_\lambda\EE\Big[\langle Q_{ll'}\rangle \langle x_i\rangle^\intercal\frac{d\lambda^{1/2}}{d\lambda_{ll'}} \lambda^{1/2} \langle x_i\rangle \Big]&=\frac12\EE_\lambda\EE\Big[\langle Q_{ll'}\rangle \langle x_i\rangle^\intercal\Big\{\frac{d\lambda^{1/2}}{d\lambda_{ll'}} \lambda^{1/2}+\lambda^{1/2}\frac{d\lambda^{1/2}}{d\lambda_{ll'}} \Big\} \langle x_i\rangle \Big]	\nn
&=\frac12\EE_\lambda\EE\Big[\langle Q_{ll'}\rangle \langle x_i\rangle^\intercal\frac{d\lambda}{d\lambda_{ll'}} \langle x_i\rangle \Big]=\EE_\lambda\EE\big[\langle Q_{ll'}\rangle \langle x_{il}\rangle\langle x_{il'}\rangle \big]\,.
\end{align*}
Combining everything in \eqref{33} and \eqref{32} yields
\begin{align*}
\sum_{l\neq l'} \EE_\lambda\EE \langle Q_{ll'}{\cal L}_{l'l}\rangle &=\sum_{l\neq l'} \EE_\lambda\EE\Big[-\langle Q_{ll'}Q_{l'l}\rangle -\langle Q_{ll'}^2\rangle +\langle Q_{ll'}\rangle\langle Q_{ll'}^{(12)}\rangle \Big]+{\cal O}_{K,S}((s_nn)^{-1/4})
\end{align*}
using 
\begin{align*}
Q^{(12)}_{ll'}\equiv \frac1n \sum_{i=1}^n x_{il}^{(1)}x_{il'}^{(2)}\quad  \text{and thus} \quad \langle Q^{(12)}_{ll'}\rangle = \frac1n \sum_{i=1}^n \langle x_{il}\rangle \langle x_{il'}\rangle\,,	
\end{align*}
the latter matrix being symmetric. The notation $A=B+{\cal O}_{K,S}((s_nn)^{-1/4})$ we introduced means simply that $|A-B|\le C(K,S)(s_nn)^{-1/4}$. We now consider the diagonal terms. Similarly
\begin{align*}
\sum_{l} \EE_\lambda\EE \langle Q_{ll}{\cal L}_{ll}\rangle &= \sum_{l} \EE_\lambda\EE\Big[ \frac{1}{2n}\sum_{i=1}^n  \langle Q_{ll}x_{il}^2\rangle -\langle Q_{ll}^2\rangle  - \frac{1}{n}\sum_{i=1}^n \Big \langle Q_{ll}x_{i}^\intercal \frac{d\lambda^{1/2}}{d\lambda_{ll}} \, {Z}_{i}  \Big\rangle\Big]\nonumber\\
&=\sum_{l}\EE_\lambda\EE\Big[   -\langle Q_{ll}^2\rangle +\frac12\langle Q_{ll}\rangle \langle Q_{ll}^{(12)}\rangle\Big]+{\cal O}_{K,S}((s_nn)^{-1/4})\,.
\end{align*}
Summing everything we obtain (note that we combine all the on and off-diagonal terms computed above inside a single double sum $\sum_{l,l'}$, and we therefore need to remove the diagonal terms counted twice)
\begin{align*}
\EE_\lambda{\rm Tr}\,\EE \langle Q{\cal L}\rangle&=\sum_{l,l'}\EE_\lambda\EE\Big[\langle Q_{ll'}\rangle \langle Q_{ll'}^{(12)}\rangle-\langle Q_{ll'}Q_{l'l}\rangle -\langle Q_{ll'}^2\rangle \Big]\nn
&\qquad\qquad\qquad+\sum_{l}\EE_\lambda\EE\Big[\langle Q_{ll}^2\rangle - \frac12 \langle Q_{ll}\rangle\langle Q_{ll}^{(12)}\rangle\Big]+{\cal O}_{K,S}((s_nn)^{-1/4})\nonumber\\
&= \EE_\lambda{\rm Tr}\,\EE\big[\langle Q\rangle\langle Q^{(12)}\rangle\big]-\EE_\lambda{\rm Tr}\,\EE\langle Q^2\rangle -\EE_\lambda\EE\big\langle\| Q \|_{\rm F}^2\big\rangle \nn
&\qquad\qquad\qquad+ \sum_{l}\EE_\lambda\EE\Big[\langle Q_{ll}^2\rangle - \frac12 \langle Q_{ll}\rangle\langle Q_{ll}^{(12)}\rangle\Big]+{\cal O}_{K,S}((s_nn)^{-1/4})\,. 
\end{align*}
Therefore, plugging this in \eqref{103easy} leads to
\begin{align*}
\EE_\lambda{\rm Tr}\,\EE \big\langle &Q({\cal L} -\EE\langle {\cal L}\rangle) \big\rangle	= -\EE_\lambda\EE\big\langle\| Q \|_{\rm F}^2\big\rangle+ \EE_\lambda\big[\| \EE\langle Q\rangle\|_{\rm F}^2\big]+ \frac12\sum_{l}\EE_\lambda\big[\EE\langle Q_{ll}^2\rangle-\EE[\langle Q_{ll}\rangle]^2\big] \nn
& +\EE_\lambda{\rm Tr}\,\EE\big\langle Q(\langle Q^{(12)}\rangle -Q)\big\rangle+  \frac12\sum_{l}\EE_\lambda\EE\big\langle Q_{ll}(Q_{ll}- \langle Q_{ll}^{(12)}\rangle)\big\rangle+{\cal O}_{K,S}((s_nn)^{-1/4})\,.
\end{align*}	
A direct application of \eqref{2terms_controlled} in Proposition~\ref{prop:thermalBOund} together with Cauchy-Schwarz gives:
\begin{align*}
\big|\EE_\lambda{\rm Tr}\,\EE\big\langle Q(\langle Q^{(12)}\rangle -Q)\big\rangle\big|\le \Big\{\EE_\lambda\EE\big\langle \|Q\|_{\rm F}^{2}\big\rangle\, \EE_\lambda\EE\big\langle \| Q-\langle Q^{(12)}\rangle\|_{\rm F}^2\big\rangle\Big\}^{1/2}	\le \frac{C(K,S)}{(s_nn)^{1/4}}
\end{align*}
as the overlap norm is bounded by $KS^2$. Similarly $$\big|\EE_\lambda\EE\big\langle Q_{ll}(Q_{ll}- \langle Q_{ll}^{(12)}\rangle)\big\rangle\big| \le \frac{C(K,S)}{(s_nn)^{1/4}}\,.$$ Therefore
\begin{align}
\EE_\lambda{\rm Tr}\,\EE \big\langle Q({\cal L} -\EE\langle {\cal L}\rangle) \big\rangle	&= -\sum_{l\neq l'}\EE_\lambda\EE\big\langle (Q_{ll'}-\EE\langle Q_{ll'}\rangle)^2\big\rangle - \frac12\sum_{l}\EE_\lambda\EE\big\langle (Q_{ll}-\EE\langle Q_{ll}\rangle)^2\big\rangle \nn
&\qquad\qquad\qquad\qquad+{\cal O}_{K,S}((s_nn)^{-1/4})\,.\label{47}
\end{align}
Finally by Cauchy-Schwarz and Proposition~\ref{prop:Lconc} we can write
\begin{align*}
 \big|\EE_\lambda{\rm Tr}\,\EE \big\langle Q({\cal L} -\EE\langle {\cal L}\rangle) \big\rangle\big| &\le 	C(S) \Big\{\EE_\lambda\EE\big\langle \|{\cal L} -\EE\langle {\cal L}\rangle\|_{\rm F}^2\big\rangle\Big\}^{1/2} \le \frac{C(C_f,K,S)}{(s_n^{4}n)^{1/6}}\,.
\end{align*}
 This inequality combined with \eqref{47} gives
 \begin{align*}
 \sum_{l\neq l'}\EE_\lambda\EE\big\langle (Q_{ll'}-\EE\langle Q_{ll'}\rangle)^2\big\rangle + \frac12\sum_{l}\EE_\lambda\EE\big\langle (Q_{ll}-\EE\langle Q_{ll}\rangle)^2\big\rangle\le \frac{C(C_f,K,S)}{(s_n^{4}n)^{1/6}}	
 \end{align*}
 and thus the final result Theorem~\ref{thm:Q_con}. \qed
\section{Concentration of the matrix ${\cal L}$}\label{app:Lconc}
The goal of this section is to prove Proposition~\ref{prop:Lconc}. The proof is broken in two parts using the decomposition
\begin{align}
\mathbb{E}\big\langle \|\mathcal{L} - \mathbb{E}\langle \mathcal{L}\rangle\|^2_{\rm F}\big\rangle
& = 
\mathbb{E}\big\langle \|\mathcal{L} - \langle \mathcal{L}\rangle\|^2_{\rm F}\big\rangle
+ 
\mathbb{E}\big[\|\langle \mathcal{L}\rangle - \mathbb{E}\langle \mathcal{L}\rangle\|^2_{\rm F}\big]\,. \label{L_decomp}
\end{align}
\subsection{Thermal fluctuations}
The first result, which is relation \eqref{thermal_L} in Proposition~\ref{prop:Lconc}, expresses concentration w.r.t.\ the posterior distribution. It follows from concavity properties of the average free energy. Recall the notation $\EE_\lambda[-]\equiv s_n^{-K(K+1)/2}\int_{{\cal D}_{n,K}} d\lambda\,[-]$.\\

\begin{proof}[Proof of \eqref{thermal_L} in Proposition~\ref{prop:Lconc}]
By direct computation we have for any $(l,l')\in\{1,\ldots,K\}^2$
\begin{align}
\mathbb{E}\big\langle (\mathcal{L}_{ll'} - \langle \mathcal{L}_{ll'} \rangle)^2\big\rangle
& = 
-\frac{1}{n}\frac{d^2f_n}{d\lambda_{ll'}^2}
+\frac{1}{n} \mathbb{E}\Big\langle\frac{d{\cal L}_{ll'}}{d\lambda_{ll'}} \Big\rangle \,.
\label{directcomputation}
\end{align}
We have shown in \eqref{bound_derL} that $$\Big|\frac{1}{n}\mathbb{E}\Big\langle \frac{d\mathcal{L}_{ll'}}{d\lambda_{ll'}} \Big\rangle\Big| \le \frac{C(S,K)}{s_nn}\,.$$
We integrate the equality \eqref{directcomputation} over 
\begin{align}\label{integrationdomain}
\lambda_{ll'}\in (a_n,b_n)=(s_n,2s_n)\ \ \text{if}\ \ l\neq l'\,, \quad \text{or}\quad  \lambda_{ll}\in (a_n,b_n)=(2Ks_n, (2K+1)s_n) \ \ \text{else}\,.	
\end{align}
We obtain
\begin{align*}
\int_{a_n}^{b_n} d\lambda_{ll'}\, \mathbb{E}\big\langle (\mathcal{L}_{ll'} - \langle \mathcal{L}_{ll'} \rangle)^2\big\rangle
& \leq 
- \frac{1}{n}\int_{a_n}^{b_n} d\lambda_{ll'} \,\frac{d^2f_{n}}{d\lambda_{ll'}^2} + \frac{C(S,K)}{n}
\nonumber \\ &
=
\frac{1}{n}\Big(\frac{df_{n}}{d\lambda_{ll'}}(\lambda_{ll'}=a_n) 
- \frac{df_{n}}{d\lambda_{ll'}}(\lambda_{ll'}=b_n)\Big)+\frac{C(S,K)}{n}\,.
\end{align*}
We have $|d f_{n}/d\lambda_{ll'}| \le |\EE[\langle x_{1l}\rangle \langle x_{1l'}\rangle]|\le S^2$ from \eqref{80}, \eqref{80_ll} so the first term is certainly smaller in absolute value than $2S^2/n$. Therefore
\begin{align*}
\frac1{s_n}\int_{a_n}^{b_n} d\lambda_{ll'}\, \mathbb{E}\big\langle (\mathcal{L}_{ll'} - \langle \mathcal{L}_{ll'} \rangle)^2\big\rangle \le \frac{2S^2}{s_nn}+\frac{C(S,K)}{s_nn}=\frac{C(S,K)}{s_nn}\,.	
\end{align*}
In the last equality, the constants $C(S,K)$ are different. We now average this inequality w.r.t. to the remaining entries of $\lambda\in{\cal D}_{n,K}$, where the set ${\cal D}_{n,K}$ is defined by \eqref{DnK}. This yields 
\begin{align*}
\EE_\lambda \mathbb{E}\big\langle (\mathcal{L}_{ll'} - \langle \mathcal{L}_{ll'} \rangle)^2\big\rangle\le \frac{C(S,K)}{s_nn}\,.	
\end{align*}
Summing all $K^2$ fluctuations for the various couples $(l,l')$ yields the desired bound. 
\end{proof}
\subsection{Quenched fluctuations}
The next proposition expresses concentration w.r.t.\ the quenched variables
and is a consequence of the concentration of the free energy onto its average (w.r.t. the quenched variables). This is where the hypothesis \eqref{hyp:f_conc} is crucial. This proposition together with \eqref{thermal_L} and relation \eqref{L_decomp} imply \eqref{total_L} in Proposition~\ref{prop:Lconc}.
\begin{proposition}[Quenched fluctuations of $\cal L$]\label{disorder-fluctuations}
	Let $(s_n)\in(0,1]^{\mathbb{N}}$ a positive sequence verifying $s_n\to 0_+$ and $s_n^4n\to+\infty$. Assume that the free energy concentrates as in identity \eqref{hyp:f_conc}. Then there exists $C(C_f,K,S)>0$ s.t.
	\begin{align*}
		\EE_\lambda \mathbb{E} \big[\|\langle \mathcal{L}\rangle-\EE\langle \mathcal{L}\rangle \|_{\rm F}^2 \big]  \le \frac{C(C_f,K,S)}{(s_n^{4}n)^{1/3}}\,.
	\end{align*}
\end{proposition}
\begin{proof}
Let us define the non-averaged free energy: $$F_{n}(\lambda)= F_{n}(\widetilde Y,Y,\theta,\lambda)\equiv-\frac{1}{n} \ln {\cal Z}_n(\widetilde Y,Y,\theta,\lambda)\,,$$ so that the averaged one given by \eqref{av_f} is simply $f_n(\lambda)= \EE\,F_{n}(\lambda)$ with $\EE=\EE_\theta\EE_{X|\theta_0}\EE_{\widetilde Y|X,\theta_{\rm out}}\EE_{Y|X,\lambda}$. We have the following identities: for any given realization of the quenched variables and any fixed $(l,l')$,
\begin{align}
 \frac{dF_{n}(\lambda)}{d\lambda_{ll'}}  &= \langle \mathcal{L}_{ll'} \rangle \,,\nonumber\\
 \frac{1}{n}\frac{d^2F_{n}(\lambda)}{d\lambda_{ll'}^2}  &= -\big\langle (\mathcal{L}_{ll'}  - \langle \mathcal{L}_{ll'} \rangle)^2\big\rangle-
 \frac{1}{n^2}\sum_{i=1}^n  \langle x_i\rangle^\intercal \frac{d^2\lambda^{1/2}}{d\lambda_{ll'}^2} Z_i \,.\label{second-derivative}
\end{align}
The same identities for $f_n(\lambda)=\EE \,F_{n}(\lambda)$ are true but with an additional average $\EE$ over the quenched variables (recall, e.g., \eqref{80}). The thermal fluctuations of $\cal L$ are almost directly equal to the the second derivative of $-\frac1n F_n(\lambda)$ as seen from \eqref{second-derivative}, up to a lower order term that we consider now. We have 
\begin{align}
\langle x_i\rangle^\intercal \frac{d^2\lambda^{1/2}}{d\lambda_{ll'}^2} Z_i=\sum_{kk'}	\langle x_{ik}\rangle \Big(\frac{d^2\lambda^{1/2}}{d\lambda_{ll'}^2}\Big)_{kk'} Z_{ik'}\le S D_n K\sum_{k'}|Z_{ik'}|\label{51}
\end{align}
where $$D_n\equiv \max_{\{(kk'),(ll')\}}\Big|\Big(\frac{d^2\lambda^{1/2}}{d\lambda_{ll'}^2}\Big)_{kk'}\Big|=C s_n^{-3/2}$$for some $C>0$. Indeed by \eqref{38} the entries of $\frac{d\lambda^{1/2}}{d\lambda_{ll'}}$ are ${\cal O}(s_n^{-1/2})$, and thus $D_n=Cs_n^{-3/2}$ by the same manipulations as for getting \eqref{38}. Therefore the following function of $\lambda_{ll'}$
\begin{align}\label{new-free}
 &F_{n,ll'}(\lambda_{ll'}) \equiv F_{n}(\lambda) - \frac{\lambda_{ll'}^2}{2} CKSs_n^{-{3/2}}\frac1n\sum_{i,k'=1}^{n,K}|Z_{ik'}|
\end{align}
is concave by construction, because thanks to 
\eqref{second-derivative}, \eqref{51} we see that the second $\lambda_{ll'}$-derivative of $F_{n,ll'}(\lambda_{ll'})$ is negative (for an appropriate positive constant $C$). Obviously its average
\begin{align}\label{new-free-2}
 f_{n,ll'}(\lambda_{ll'}) \equiv \mathbb{E} \,F_{n,ll'}(\lambda_{ll'})= f_{n}(\lambda)  -\frac{\lambda_{ll'}^2}{2} CKSs_n^{-{3/2}}\frac1n\sum_{i,k'=1}^{n,K}\EE\,|Z_{ik'}|
\end{align}
is concave too.
Concavity then allows to use the following standard lemma (see \cite{BarbierM17a,barbier2017phase} for a proof):
\begin{lemma}[A bound for concave functions]\label{lemmaConvexity}
Let $G(x)$ and $g(x)$ be concave functions. Let $\delta>0$ and define $C^{-}_\delta(x) \equiv g'(x-\delta) - g'(x) \geq 0$ and $C^{+}_\delta(x) \equiv g'(x) - g'(x+\delta) \geq 0$. Then
\begin{align*}
|G'(x) - g'(x)| \leq \delta^{-1} \sum_{u \in \{x-\delta,\, x,\, x+\delta\}} |G(u)-g(u)| + C^{+}_\delta(x) + C^{-}_\delta(x)\,.
\end{align*}
\end{lemma}
First, from \eqref{new-free}, \eqref{new-free-2} we have 
\begin{align}\label{fdiff}
 F_{n,ll'}(\lambda_{ll'}) - f_{n,ll'}(\lambda_{ll'}) = F_{n}(\lambda) - f_{n}(\lambda) - \frac{\lambda_{ll'}^2}{2} CKSs_n^{-{3/2}} A_n  
\end{align} 
with
\begin{align*}
 A_n \equiv \frac{1}{n}\sum_{i,k'=1}^{n,K} (\vert Z_{ik'}\vert -\mathbb{E}\,\vert  Z_{ik'}\vert)\,.	
\end{align*}
Second, we obtain for the $\lambda_{ll'}$-derivatives difference
\begin{align}\label{derdiff}
 F_{n,ll'}'(\lambda_{ll'}) - f_{n,ll'}'(\lambda_{ll'}) = 
\langle \mathcal{L}_{ll'} \rangle-\mathbb{E}\langle \mathcal{L}_{ll'} \rangle -\lambda_{ll'}CKSs_n^{-3/2}A_n \,.
\end{align}
From \eqref{fdiff} and \eqref{derdiff} it is then easy to show that Lemma \ref{lemmaConvexity} applied to $G(x) \to F_{n,ll'}(\lambda_{ll'})$ and $g(x)\to f_{n,ll'}(\lambda_{ll'})$ with $x\to \lambda_{ll'}$ gives
\begin{align}\label{usable-inequ}
\vert \langle \mathcal{L}_{ll'}\rangle - \mathbb{E}\langle \mathcal{L}_{ll'}\rangle\vert&\leq 
\delta^{-1} \sum_{u\in \{\lambda_{ll'} -\delta,\, \lambda_{ll'},\, \lambda_{ll'}+\delta\}}
 \Big(\vert F_{n}(u) - f_{n}(u) \vert + \frac{u^2}{2}CKSs_n^{-3/2}\vert A_n \vert \Big)\nonumber\\
 &\qquad\qquad\qquad\qquad\qquad
  + C_\delta^+(\lambda_{ll'}) + C_\delta^-(\lambda_{ll'}) + CKS(2K+1)s_n^{-1/2}\vert A_n\vert
\end{align}
where we used $|\lambda_{ll'}|\le (2K+1)s_n$ for any $(l,l')\in\{1,\ldots,K\}^2$ as $\lambda\in {\cal D}_{n,K}$, and where $$C_\delta^-(\lambda_{ll'})\equiv f_{n,ll'}'(\lambda_{ll'}-\delta)-f_{n,ll'}'(\lambda_{ll'})\ge 0\,, \qquad C_\delta^+(\lambda_{ll'})\equiv f_{n,ll'}'(\lambda_{ll'})-f_{n,ll'}'(\lambda_{ll'}+\delta)\ge 0\,.$$ Note that $\delta$ will be chosen small enough so that when $\lambda_{ll'}$ is varied by $\pm \delta$ the matrix $\lambda$ remains in the set ${\cal D}_{n,K}$. Remark that by independence of the noise variables $$\mathbb{E}[A_n^2]  \le (1-2/\pi)\frac{K}{n}\le \frac{K}n\,.$$ We square the identity \eqref{usable-inequ} and take its expectation. Then using $(\sum_{i=1}^pv_i)^2 \le p\sum_{i=1}^pv_i^2$ by convexity, and again that $|\lambda_{ll'}|\le (2K+1)s_n$ as well as the free energy concentration hypothesis~\eqref{hyp:f_conc} (irrelevant positive numerical constants are absorbed in the generic positive constants $C$),
\begin{align}\label{intermediate}
 \mathbb{E}\big[(\langle \mathcal{L}_{ll'}\rangle - \mathbb{E}\langle \mathcal{L}_{ll'}\rangle)^2\big]
 &
 \leq \, 
  \frac{C_f}{n\delta^2} + ((2K+1)s_n+\delta)^4 \frac{CK^3S^2}{n\delta^2s_n^3}   + C_\delta^+(\lambda_{ll'})^2 + C_\delta^-(\lambda_{ll'})^2
 + \frac{CK^4S^2}{s_nn} \, .
\end{align}
Recall $|C_\delta^\pm(\lambda_{ll'})|=|f'_{n,ll'}(\lambda_{ll'}\pm\delta)-f'_{n,ll'}(\lambda_{ll'})|$. We have (here $Z\sim{\cal N}(0,1)$)
\begin{align}
|f'_{n,ll'}(\lambda_{ll'})| = \big|\EE\langle {\cal L}_{ll'}\rangle -\lambda_{ll'}CK^2Ss_n^{-3/2}\EE\,|Z|\big|  \leq S^2+CK^2Ss_n^{-1/2}\le C(S)K^2s_n^{-1/2}\label{boudfprime}	
\end{align}
using \eqref{80}, \eqref{80_ll} to assess that $|\EE\langle {\cal L}_{ll'}\rangle|\le S^2$, and that $\lambda_{ll'}$ is propositional to $s_n$ because $\lambda \in{\cal D}_{n,K}$. Thus we have the crude bound $$|C_\delta^\pm(\lambda_{ll'})|\le C(S)K^2s_n^{-1/2}\,.$$ Consider again the integration domain \eqref{integrationdomain}. We reach
\begin{align*}
 \int_{a_n}^{b_n} d\lambda_{ll'}\,&\big\{C_\delta^+(\lambda_{ll'})^2 + C_\delta^-(\lambda_{ll'})^2\big\}\leq  C(S)K^2s_n^{-1/2}
 \int_{{a_n}}^{{b_n}} d\lambda_{ll'}\, \big\{C_\delta^+(\lambda_{ll'}) + C_\delta^-(\lambda_{ll'})\big\}
 \nonumber \\ 
&= 
C(S)K^2s_n^{-1/2}\big[\big(f_{n,ll'}(a_n+\delta) - f_{n,ll'}(a_n-\delta)\big)+ \big(f_{n,ll'}(b_n-\delta) - f_{n,ll'}(b_n+\delta)\big)\big]\,.
\end{align*}
The mean value theorem and \eqref{boudfprime} 
imply $$|f_{n,ll'}(\lambda_{ll'}-\delta) - f_{n,ll'}(\lambda_{ll'}+\delta)|\le \frac{C(S)K^2\delta}{\sqrt{s_n}} \,.$$ Therefore
\begin{align*}
 \int_{a_n}^{b_n} d\lambda_{ll'}\, \big\{C_\delta^+(\lambda_{ll'})^2 + C_\delta^-(\lambda_{ll'})^2\big\}\leq  C(S)K^4\frac{\delta}{s_n}\,.
\end{align*}
Averaging \eqref{intermediate} over $\lambda_{ll'}\in (a_n, b_n)$ and choosing $\delta=\delta_n$ s.t. $\delta_n/s_n\to 0_+$ yields
\begin{align*}
 &\frac{1}{s_n}\int_{a_n}^{b_n} d\lambda_{ll'}\, 
 \mathbb{E}\big[(\langle \mathcal{L}_{ll'}\rangle - \mathbb{E}\langle \mathcal{L}_{ll'}\rangle)^2\big]\leq \frac{C_f}{n\delta_n^2}+CK^7S^2\frac{s_n}{n\delta_n^2} +C(S)K^4\frac{\delta_n}{s_n^2}+ \frac{CK^4S^2}{s_nn}\,.
\end{align*}
We optimize the bound by choosing $\delta_n  = s_n^{2/3}n^{-1/3}$ (which indeed verifies $\delta_n/s_n\to 0$ as long as $s_nn\to +\infty$). The desired result is then obtaind after averaging over the remaining $K(K+1)/2-1$ independent entries of $\lambda\in{\cal D}_{n,K}$, and summing the $K^2$ fluctuations for the various $(l,l')$ couples.	
\end{proof}

\section*{Acknowledgments}
Funding from the ``Chaire de recherche sur les mod\`eles et sciences des donn\'ees'', Fondation CFM pour la Recherche-ENS is acknowledged. I would like to deeply thank Nicolas Macris, Antoine Maillard and Florent Krzakala for stimulating exchanges. In particular Nicolas Macris for the very many discussions leading to the idea of the perturbation described in section~\ref{sec:Gauss_channel}, his careful reading of the manuscript, but most of all for his constant support. I am also very grateful to L\'eo Miolane for pointing identity \eqref{Leo_0}, and to Cl\'ement Luneau for his careful reading and corrections of the manuscript. Last but not least, I want to express all my gratitude to Dmitry Panchenko for helping me facing conceptual issues on related questions, and his detailed and patient explanations on some of his works.

\bibliographystyle{./unsrt_abbvr.bst}
{ \bibliography{refs_2}}

\end{document}